\documentclass{article}

\usepackage{rotating}
\usepackage{multirow}
\usepackage{subfigure}
\usepackage{verbatim}
\usepackage{algorithm}
\usepackage{algorithmic}
\usepackage{tikz}
\usepackage{pgf}
\usepackage{pgfplots}
\usepackage{amsmath}
\usepackage{amssymb}
\usepackage{makecell}

\newtheorem{definition}{Definition}
\newtheorem{proposition}{Proposition}
\newtheorem{proof}{Proof}

\usepackage[textwidth=6.5in, textheight=8.7in, top=1.2in, left=1.2in]{geometry}

\usepackage[utf8]{inputenc}

\usepackage{natbib}
 \bibpunct[, ]{(}{)}{,}{a}{}{,}%
 \def\newblock{\ }%
\begin{document}

\author{Sophie N. Parragh \\ Fabien Tricoire\\  
  {\small \it  Institute of Production and Logistics Management}\\
  {\small \it Johannes Kepler University, Linz, Austria} \\
  {\small \it \{sophie.parragh,fabien.tricoire\}@univie.ac.at}
}

\title{Branch-and-bound for bi-objective integer programming
  \footnote{Accepted for publication in INFORMS Journal On Computing; doi: 10.1287/ijoc.2018.0856}
}
  
\maketitle

\begin{abstract}%
In
bi-objective integer optimization the optimal result corresponds to a
set of non-dominated solutions. We propose a generic bi-objective
branch-and-bound algorithm that uses a
problem-independent branching rule exploiting available integer solutions
and takes 
advantage of integer objective coefficients.
The
developed algorithm is applied to bi-objective facility location problems, to
the bi-objective set covering problem, as well as to the bi-objective team
orienteering problem with time windows. In the latter case, lower bound sets are computed by means of column generation.
Comparison to state-of-the-art exact algorithms
shows the effectiveness of the proposed branch-and-bound algorithm.
\end{abstract}%

\section{Introduction}\label{sec:introduction}
We live in a world full of trade-offs. In almost every problem situation we
encounter, it is difficult to define the one and only goal to aim for,
especially whenever more than one decision maker or stakeholder is
involved. Thus, many if not all practical problems involve several and often
conflicting objectives. Prominent examples are profitability or cost versus
environmental concerns~\citep{ramos2014planning} or customer satisfaction
\citep{braekers2015bi}.
Whenever it is not possible to aggregate these conflicting objectives in a meaningful way, it is recommendable to produce the 
set of trade-off solutions in order to elucidate their trade-off relationship.

In this paper, we propose a branch-and-bound framework which produces the optimal set of trade-off solutions for optimization problems with two objectives which can be modeled as integer linear programs.
The main contributions are 
a new problem-independent branching rule for
bi-objective optimization, building up on the Pareto branching idea of \citet{Stidsen:2014}, and several improvements exploiting the integrality of
objective function values.
We compare our framework to existing general purpose state-of-the-art frameworks, such as the $\epsilon$-constraint framework and the balanced box method, and to results obtained by means of another branch-and-bound algorithm developed by \citet{Belotti:2013}.
Furthermore, we apply it to a bi-objective orienteering problem for which lower bound sets are produced
using column generation.
Although column generation has been used in the context of bi-objective optimization, e.g. in the context of robust airline crew scheduling~\citep{tam11}, vehicle routing~\citep{sarpong13}, and bi-objective linear programming~\citep{raith12}, 
this is, to the best of our knowledge, the first time column generation is used in a bi-objective branch-and-bound algorithm, resulting in a bi-objective branch-and-price
algorithm.

The paper is organized as follows. In Sections~\ref{sec:preliminaries} and~\ref{sec:relatedWork}, we give preliminaries and we review related work. In Section~\ref{sec:bobab}, we present the key ingredients to our branch-and-bound framework and in Section~\ref{sec:integralityImprovements} we discuss the proposed enhancements. Experimental results and how we embed column generation into our branch-and-bound framework are presented in Section~\ref{sec:exp}. Section~\ref{sec:conclusion} concludes the paper.

\section{Preliminaries}
\label{sec:preliminaries}
The aim of this paper is to solve bi-objective integer linear programs of the following form:\\
\vspace{-5ex}
\begin{align}
	\min_x \quad  &f(x) = (f_1(x), f_2(x) ) \nonumber \\
	\text{s.t.} \quad 	& Ax = b  \label{IP}\\
					& x_i \quad \text{integer} && \forall i = 1, .\ldots , n \nonumber
\end{align}
where $f(x)$ is the vector of objective functions, $A$ is the constraint matrix of size $m \times n$, $b$ the right hand side vector of size $m$, and $x_i$ the decision variables which are restricted to integer values. An important subclass are binary (combinatorial) problems. 

In formulation~\eqref{IP} we strive for identifying the set of Pareto optimal
solutions according to the given objectives. A feasible solution is Pareto
optimal if there does not exist any feasible solution which dominates it, i.e. which is better in one objective and not worse in any other objective. The image of a Pareto optimal solution in criterion space is referred to as non-dominated point. Each non-dominated point may represent more than one Pareto optimal solution. We aim at generating all non-dominated points, i.e. at least one Pareto optimal or efficient solution per non-dominated point.

Depending on the structure of the objective functions and the feasible domains
of the variables, the set of all non-dominated points
may take different shapes. In the case
where all variables assume continuous values, it
corresponds to the boundary of the convex hull of the image of all feasible
solutions in criterion space.  It can be characterized by the extreme points of
the boundary. These extreme points can be computed efficiently, e.g. by the
\emph{weighted sum method} of \citet{Aneja:1979}. It combines the two
objectives into a single one, resulting in $(w_1 f_1(x) + w_2 f_2(x))$, with $w_1$ and $w_2$ giving the weights of objective one and two, respectively.
These weights are
changed systematically so as to obtain the entire set of extreme points.  In the
case where all variables may only take integer values, which we look at, the
set of non-dominated points is composed of points which may or may not lie on the
boundary of the convex hull in criterion space. Solutions whose images in
criterion space correspond to extreme points are called \emph{extreme
  (supported) efficient solutions}, those that correspond to points that lie on
the boundary are called \emph{supported efficient solutions} and those that
correspond to points that lie in the interior of the convex hull but are still
non-dominated are called \emph{non-supported efficient solutions}.
In the integer case, the set of extreme points can also be computed by means of the \emph{weighted sum method}.

Additional important notions in multi-objective optimization are \emph{ideal point} and \emph{nadir point}.
The ideal point is constructed by combining the best possible values for each objective function. It dominates all non-dominated points. The nadir point is its opposite, a point that is dominated by all non-dominated points. It is constructed by combining the worst possible values for each objective function of all efficient solutions.
Both notions have been used in the context of branch-and-bound algorithms for multi-objective optimization.
For a more detailed introduction to multicriteria decision making, we refer to \cite{ehrgott2005multicriteria}.

\section{Related work}
\label{sec:relatedWork}

Exact approaches in multi-objective (mixed)
integer programming can be divided into two classes: those that work in the
space of objective function values (referred to as criterion space search
methods, e.g. by \citet{Boland:2015}) and those that work in the space of 
decision variables
(generalizations of branch-and-bound algorithms). 

Criterion space search methods solve a succession of single-objective problems in order to compute the set of Pareto optimal solutions. Therefore, 
they exploit the power of state-of-the-art mixed integer programming solvers.
This appears to be one of their main advantages in comparison to generalizations of branch-and-bound algorithms \citep{Boland:2013a}. However, many combinatorial single-objective problems cannot be solved efficiently by commercial solvers, e.g. the most efficient exact algorithms in the field of vehicle routing rely on column generation based techniques \citep[see e.g.][]{Baldacci:2012survey}. Thus, especially in this context but also in general, it is not clear whether a criterion space search method or a bi-objective branch-and-bound algorithm is more efficient.

One of the most popular criterion space search methods is the \emph{$\epsilon$-constraint method}, first introduced by~\cite{Haimes:1971}. It
consists in iteratively solving single-objective versions of a bi-objective
problem. In every step, the first objective is optimized, but a constraint is
updated in order to improve the quality of the solution with regards to
the second objective. Thus all non-dominated points are enumerated.
The \emph{$\epsilon$-constraint method} is generic and simple to implement and it is among the best performing criterion space search algorithms when applied, e.g. to the bi-objective prize-collecting Steiner tree problem~\citep{leitner14}.

Recently, the \emph{balanced box method} for bi-objective 
integer programs has been introduced by \citet{Boland:2015}. Optimal solutions
are computed for each objective and they
define a rectangle. This rectangle is then split in half and in each half, the
objective which has not been optimized yet in this half of the rectangle is
then optimized. At this stage there are two non-overlapping rectangles. The same procedure is repeated recursively on
these rectangles.

Generalizations of branch-and-bound to
multiple objectives for mixed 0-1 integer programs are provided by \cite{Mavrotas:1998}. Their bounding procedure
considers an ideal point for the upper bound (they work on a maximization
problem), consisting of the best possible 
value for each objective at this node, and keeps branching until this ideal
point is dominated by the lower bound set. \cite{Vincent:2013} improve the
algorithm by~\cite{Mavrotas:1998}, most notably by comparing bound sets instead
of ideal points in the bounding procedure.
\citet{masin2008} propose a branch-and-bound method that uses a surrogate objective function returning a single numerical value. This value can be treated like a lower bound in the context of single objective minimization and conveys information on whether the node can be fathomed or not. They illustrate their method using a three-objective scheduling problem example but no computational study is provided.
\cite{Sourd:2008} use separating hyperplanes between upper and lower bound sets
in order to discard nodes in a general branch-and-bound framework for integer programs. The concept is in fact similar to the lower bound sets defined by~\cite{Ehrgott:2006}.
\cite{Jozefowiez:2012} introduce a branch-and-cut algorithm for integer
programs in which discrete sets are used for lower bounds, so nodes can be
pruned if every point in the lower bound set is dominated by a point in the
upper bound set. Partial pruning is used in order to discard parts of the
solutions represented by a given node of the search tree, thus speeding up the
whole tree search.
\cite{Belotti:2013} are the first to introduce a general-purpose
branch-and-bound algorithm for bi-objective mixed integer linear
programming, where the continuous variables may appear in both objective
functions.
They build up on the
previous work by~\cite{Visee:1998}, \cite{Mavrotas:1998}, \cite{Sourd:2008} and
\cite{Vincent:2013} and like them, they use a binary branching scheme. Improved fathoming rules are introduced in order to
discard more nodes during tree search.  \cite{adelgren14} propose an efficient data structure to store and update upper bound sets in the context of mixed integer programming and illustrate its efficiency using the algorithms of \cite{Belotti:2013} and of \cite{adelgren2016BBMOMIP}.

Ideas most related to ours are presented by
\citet{Stidsen:2014}. They provide a branch-and-bound algorithm to deal with
a certain class of bi-objective mixed integer linear programs. They propose two binary branching schemes that exploit
partial dominance of already visited integer solutions and the current upper
bound set, respectively. Our approach relies on bound sets to represent the
lower bound, in contrast to \citet{Stidsen:2014} who solve one scalarized
single objective problem in each node of their branch-and-bound tree. In
addition, in our scheme, the number of child nodes depends on the number of
currently non-dominated portions which are obtained after applying what we call
\emph{filtering} of the lower bound set by the current upper bound set. Ideas
similar to the Pareto branching concept of \citet{Stidsen:2014} and its
extension to bound sets presented in this paper, which we name objective space
branching, are also being explored by \citet{adelgren2016BBMOMIP} in the
context of a branch-and-bound method for multi-objective mixed integer
programming and  by \citet{Gadegaard2016:workingPaper} in the context of a bi-objective branch-and-cut algorithm.

Furthermore, we also propose several improvements exploiting the integrality of objective
functions. These improvements, which we call \emph{segment tightening}, \emph{lower bound lifting}, and  \emph{integer dominance} allow to disregard even larger regions of the
objective space. We illustrate their impact on the overall performance of the algorithm in the experimental results section.

\section{A branch-and-bound framework for bi-objective integer
  optimization}\label{sec:bobab}

Branch-and-bound is a general purpose tree search method to solve (mixed)
integer linear programs. As its name suggests, it has two main ingredients:
branching and bounding. Branching refers to the way the search space is
partitioned, while bounding refers to the process of using valid bounds to
discard subsets of the search space without compromising optimality. In the
following, we first describe the general framework of our bi-objective
branch-and-bound (BIOBAB). We then describe its different ingredients in
further detail.

\subsection{General idea}
\label{sec:genidea}
Our BIOBAB is similar to a
traditional single-objective branch-and-bound algorithm, except for the fact
that we compare bound sets, instead of single numerical values.

In the beginning the upper bound (UB) set is empty. It is  
updated every time an integer solution $x$ is obtained. One of
three things can happen:
\begin{enumerate}
\item if $x$ is dominated by at least one solution
from the UB set, then $x$ is discarded,
\item if some solutions from the UB set are dominated by $x$, then these
  solutions are discarded and $x$ is added to the UB set,
\item if $x$ neither dominates nor is dominated by any solution from the UB set,
$x$ is directly added to the UB set.
\end{enumerate}
These conditions can be tested in linear time. If the UB set is sorted
according to the values of one objective, it is possible to check dominance in
logarithmic time using binary search.

In each branch-and-bound node, lower bound (LB) sets are computed by
means of an algorithm similar to the \emph{weighted sum method} of
\cite{Aneja:1979}, described in detail in Section~\ref{sec:LBlifting}, with either the linear relaxation or the original integer program
(IP), since both yield valid bound sets  \citep{Ehrgott:2006}.

A branch-and-bound node contains information on branching decisions as usual, but
also a bound for each objective. This is used to partition the objective space. The
objective bounds at the root node can be known valid bounds, or infinite
values.

Once the LB set at a given node has been calculated, we apply what we call
\emph{filtering} (see Section~\ref{sec:filter}): the current UB set and the current
LB set are compared
and a set of non-dominated portions of the current LB set are returned. In the
case where this set is empty, the node can be fathomed. Otherwise branching may
be required; the filtered LB set is then used as input for the branching
procedure (see Section~\ref{sec:branching}). If the LB set
corresponds to a leaf (i.e. the node corresponds to a unique integer feasible
solution), or if it is completely dominated following filtering, then the node
can be fathomed and branching is not necessary. Otherwise a number of child
nodes are generated, which are added to the list of branch-and-bound nodes to
be processed. Once all nodes have been explored, the search is over and the
complete Pareto set for the original problem is contained in the UB set.

We now present the concept of \emph{lower bound segment} which is used in the bounding, filtering
and branching procedures and which is also key to the proposed improvements discussed in 
Section~\ref{sec:integralityImprovements}. 

\subsection{Lower bound segments}
 In the following we assume 
that a LB set has been produced.
As mentioned above, it is described by
the 
extreme points of the convex hull (of the image of
the feasible set in objective space) for the linear
relaxation of the problem (or the original IP).

To any two consecutive extreme points $p=(p_1,p_2)$ and $q=(q_1,q_2)$ we can 
associate the subset of objective space that is \emph{covered} or dominated and
that represents the possibility to find feasible solutions 
to the
integer problem.
To define this covered space, we associate to each pair of extreme points a
top-right corner (or local nadir) point $c= (c_1,c_2)$ 
of which coordinates are valid
single-objective upper bounds, one per objective, on the space dominated by the
two extreme points and the line connecting them.

Let $\Xi$ be the set of all points in the objective space which are
associated to feasible solutions to the bi-objective integer problem at hand.
We now formally introduce the concept of \emph{lower bound segment} that we use
to represent LB sets.

\begin{definition}
 Three points $p=(p_1,p_2), q=(q_1,q_2)$, and $c=(c_1,c_2)$, such that 
  $p_1 < q_1 \land p_2 > q_2$, 
  define a \emph{lower bound segment}  $s$ iff
  $\{(z_1, z_2) \in \Xi | z_2 < a z_1 + b\} = \emptyset$,
  where $a = (q_2 - p_2) / (q_1 - p_1)$ is the slope of the line
  defined by the two points and $b = p_2 - a \cdot p_1$ is the y-intercept of
  this same line; $c$ is a valid local nadir point, such that $c_1 \geq
  q_1$ and $c_2 \geq p_2$.
\end{definition}
Thus, any point $z = (z_1, z_2) \in \Xi$ is on or above the line defined by
$p$ and $q$.
Note that we do not include the case where $p = q$ into our
definition. For convenience, we extend the definition of segments to allow
equal points:

\begin{definition}
  Two points $p = (p_1, p_2)$ and $c = (c_1, c_2)$
  define a \emph{lower bound segment}  $s$ iff
  there does not exist a feasible point $z$ that dominates $p$; 
  $c$ is a valid local nadir point, such that $c_1 \geq p_1$ and $c_2
  \geq p_2$
  \label{def:onepointsegment}
\end{definition}

Figure~\ref{fig:segment-dominated-space} depicts such a lower bound segment
(shaded).
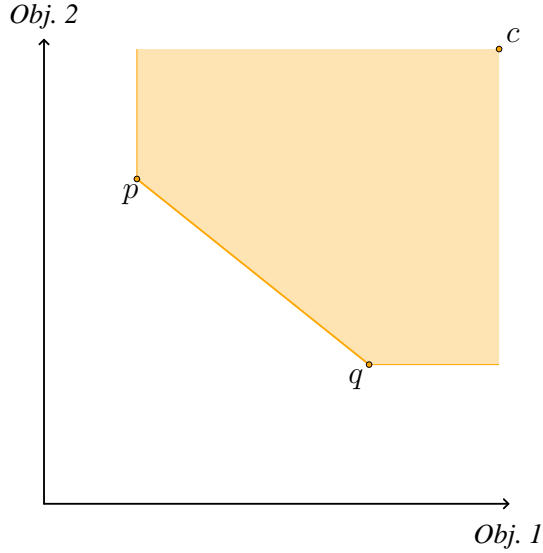
\begin{figure}
  \begin{center}
      \begin{tikzpicture}
        \node[anchor=south west,inner sep=0] (image) at (0,0) {\includegraphics[scale=.35]{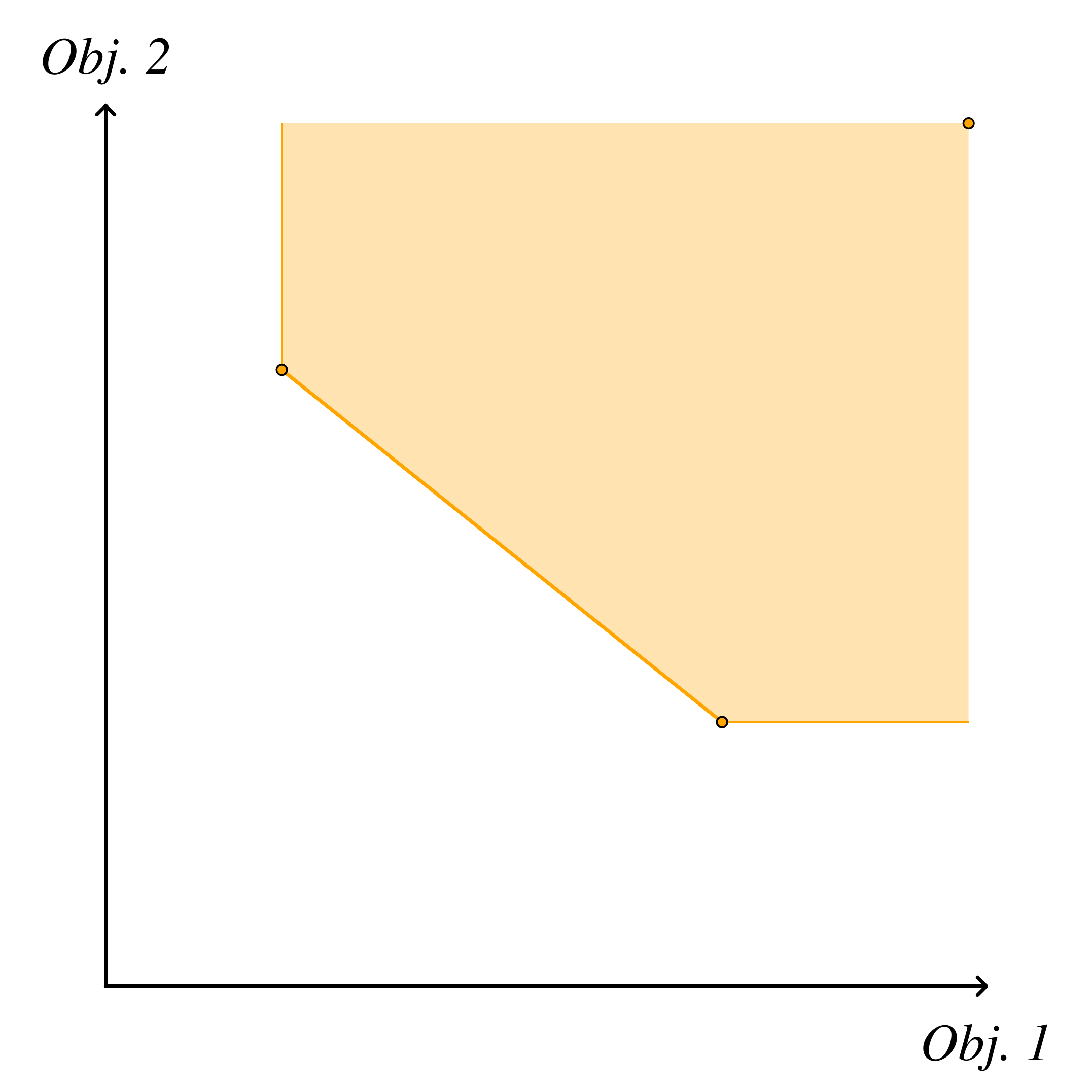}};
        \begin{scope}[x={(image.south east)},y={(image.north west)}]
          \large
           \node[anchor=north east] at (0.28,0.67) {$p$};
          \node[anchor=north east] at (0.67,0.35) {$q$};
          \node[anchor=south west] at (0.88,0.88) {$c$};
        \end{scope}
      \end{tikzpicture}
    \caption{Lower bound segment defined by the two extreme points $p$
     and $q$, and local nadir point $c$.}
    \label{fig:segment-dominated-space}
  \end{center}
\end{figure}
In the following we use object-oriented notation to refer to the
attributes defining a given segment: if extreme points $u$, $v$ and local nadir
point $w$ define a segment $s$, then, $s.p = u$, $s.q = v$ and $s.c =
w$. Additionally, $s.a$ is the slope of the line going through $s.p$ and
$s.q$, while $s.b$ is its $y$-intercept.
Graphically, a segment $s$ is a convex polygon defined by the points $s.p$,
$s.q$, $(s.c_1,s.q_2)$, $s.c$, $(s.p_1,s.c_2)$.

In the most basic case the coordinates of the  local nadir point can be arbitrarily large values.
If the two extreme points of the set of non-dominated points
are known, we can deduce a
tighter valid local nadir point from them. Similarly, we can associate to any
LB set a valid local nadir point by considering, for each objective, the
maximum value among all local nadir points.

Here we note that we can partition the objective
space, using values for any of the two objectives,
or any linear combination of both with positive weights~\citep[as is done
in][]{Stidsen:2014}.
In our BIOBAB, we partition the objective space using
different intervals for the first
objective to split the LB set into vertical stripes.

In practice, this allows us to improve the initial local nadir point for
segments of a connected sequence of segments: for any two segments connected by
point $u = (u_1, u_2)$, we can cut from the space covered by the left
segment all points $z = (z_1, z_2)$ such that $z_1 \geq u_1$, since these
points are also covered by the right segment.
By doing so, we partition the objective space
covered by the LB set into different LB segments with
associated local nadir points.

UB sets can also be used to provide better local nadir
points. For instance, following our 
previous example from Figure~\ref{fig:segment-dominated-space}, consider point
$u$, being the image of a feasible solution in objective space, such that
$u_1 \geq q_1$ and $u_2 < q_2$,
then $u_1$ defines a valid bound on the first objective
of the nadir point, 
since any point $v$
such that $v_1 \geq u_1$ and $v_2 \geq q_2$ is dominated by $u$. This
example
is described graphically in Figure~\ref{fig:tighten-boundright}: for the depicted segment, using UB point $u$, the initial local nadir point $c$ can be improved to $c’$.
\begin{figure}
  \begin{center}
      \begin{tikzpicture}
        \node[anchor=south west,inner sep=0] (image) at (0,0) {\includegraphics[scale=.35]{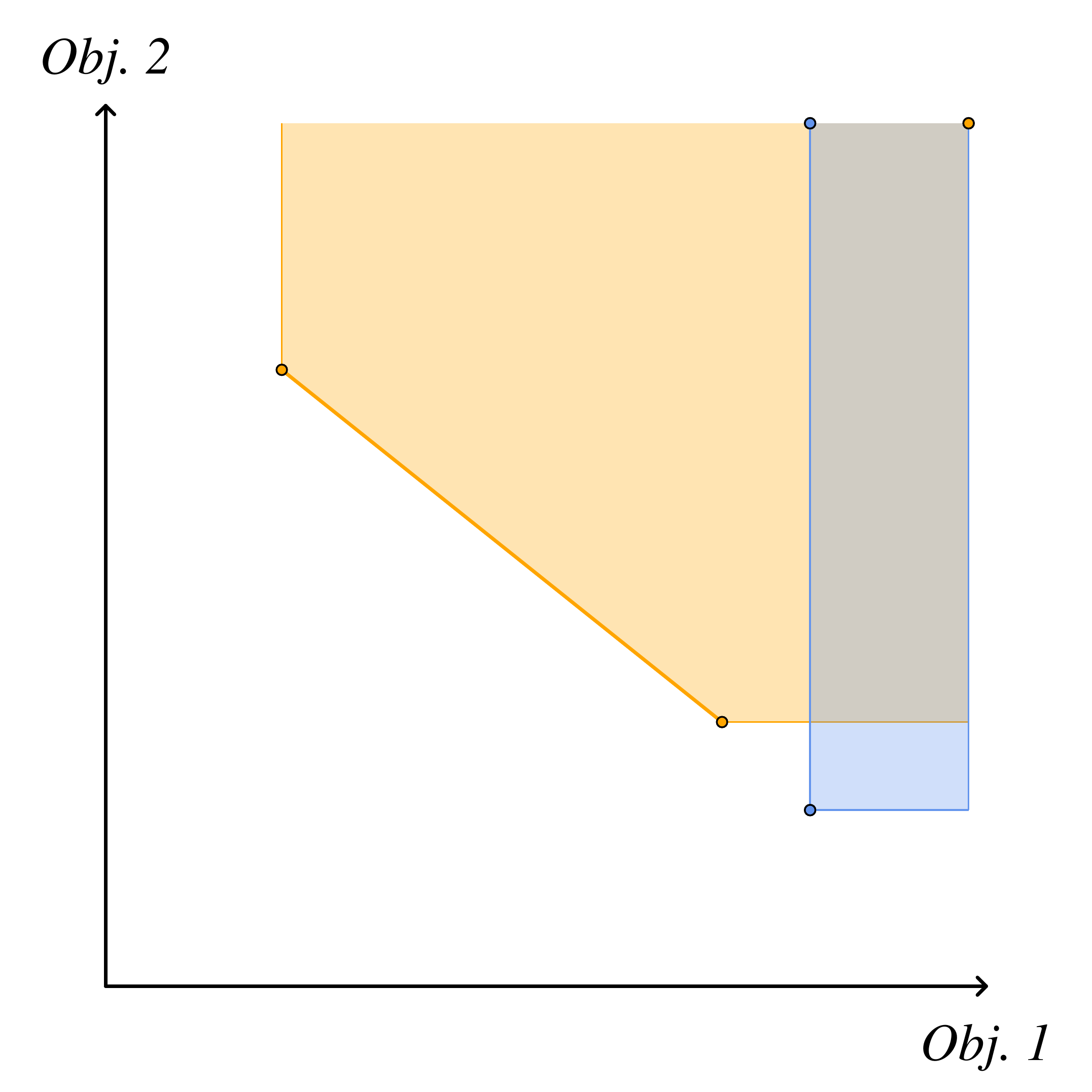}};
        \begin{scope}[x={(image.south east)},y={(image.north west)}]
          \large
           \node[anchor=north east] at (0.28,0.67) {$p$};
          \node[anchor=north east] at (0.67,0.35) {$q$};
          \node[anchor=south west] at (0.88,0.88) {$c$};
          \node[anchor=south west] at (0.72,0.88) {$c'$};
          \node[anchor=north east] at (0.76,0.26) {$u$};
        \end{scope}
      \end{tikzpicture}
    \caption{Space dominated by the segment between points $p$ and
      $q$ can be reduced by removing from it the space dominated by
     upper bound point $u$. Therefore local nadir point $c$ can be improved to
    $c'$}
    \label{fig:tighten-boundright}
  \end{center}
\end{figure}
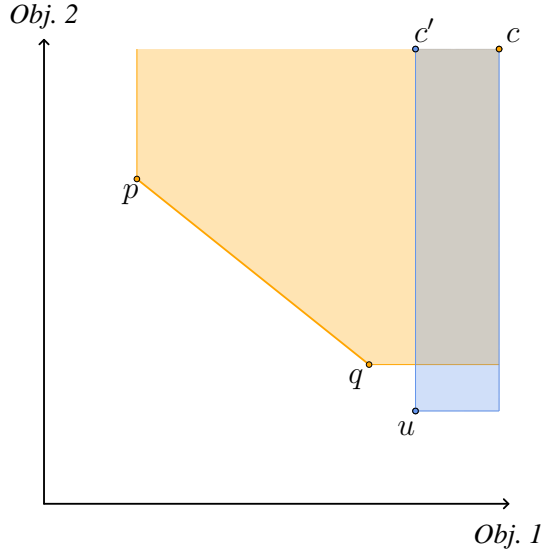

The fact that LB sets can be reduced using UB sets is the basis for a
bi-objective branching rule that is presented in
Section~\ref{sec:branching}.

\subsection{Bound set filtering and node fathoming \label{sec:filter}}

One major difficulty in previous bi-objective branch-and-bound
approaches lies in the evaluation of the dominance of a given LB set
by a given UB set. Multiple fathoming rules have been developed over
the years \citep[see][for the current state of the art]{belotti2016fathoming}. We now
introduce the fathoming rule used in BIOBAB.

As seen above, it is possible to split a LB set into LB segments (with
associated local nadir
points). If we establish that each segment of a
current LB set is dominated by the UB set, then we also establish that the
whole LB set is dominated, therefore the current node of the
branch-and-bound tree can be fathomed. In order to determine if a LB segment is
dominated by an UB set, we simply subtract the space dominated by the UB set
from the LB segment:
if the remaining space is empty then
the LB segment is dominated, otherwise this segment may be tightened (as shown
for example in
Figure~\ref{fig:tighten-boundright}).

Algorithm~\ref{alg:filter-segment}
details how to subtract the space covered by an UB point $u = (u_1, u_2)$ from a LB segment
$s$. It shares some similarities with the dominance rules used
by~\cite{Vincent:2013}, but considers more cases. Notably, it updates the local
nadir point even in cases where $u$ does not dominate any point on the line connecting 
$s.p$ and $s.q$.
We use the object-oriented notation defined above.
A pair $(z_1, z_2)$ represents a point and $Segment(p, q, c)$ is used to create
a segment out of points $p, q$ and $c$.
Lines 2-7 deal with the case where $u$ does not dominate any point on the line connecting 
$s.p$ and $s.q$.
In that case, $u$ may still be above $s.p$ or to the right of
$s.q$, in which cases there is potential for improving $s.c$. If $u$ is to
the right of $s.q$ and below $s.q$ (Lines 5-7), we
are in the situation depicted in Figure~\ref{fig:tighten-boundright} and
$s.c_1$ can be reduced to $min(s.c_1, u_1)$. Similarly, if $u$ is above
$s.p$ and to the left of $s.p$ (Lines 3-4), $s.c_2$ can then be reduced to
$min(s.c_2, u_2)$.
Lines 8-15 deal with all the cases when at least one part of the line
connecting $s.p$ and $s.q$ is dominated
by $u$. There are four different cases: (i) $u$ does not dominate $s.p$
(Lines 9-11), (ii) $u$ does not dominate $s.q$ (Lines 12-14),
(iii) $u$ dominates neither $s.p$ nor $s.q$, resulting in two new
segments (both clauses in Lines 9-11 and 12-14 are matched), and (iv) $u$
dominates both $s.p$ and $s.q$. In case (iv) no clause is matched in the
algorithm, therefore an empty set is returned.

\begin{algorithm}
  \caption{$filterSegment(s, u)$: filter LB segment $s$ using UB point $u = (u_1, u_2)$}
  \begin{algorithmic}[1]
    \STATE $S \gets \emptyset$
    \IF{$u_2 \geq s.a \cdot u_1 + s.b \lor u_2 \geq s.p_2 \lor u_1 \geq
      s.q_1$}
    \IF{$u_1 \leq s.p_1$}
    \STATE $S \gets S \cup \{Segment(p, q, (c_1, min(c_2, u_2)))\}$
    \ELSIF{$u_2 \leq s.q_2$}
    \STATE $S \gets S \cup \{Segment(p, q, (min(c_1, u_1), c_2))\}$
    \ENDIF
    \ELSE
    \IF{$u_1 > s.p_1$}
    \STATE $S \gets S \cup \{Segment(p, (u_1, s.a \cdot  u_1 + s.b),
    (u_1, c_2))\}$
    \ENDIF
    \IF{$u_2 > q_2$}
    \STATE $S \gets S \cup \{Segment(( (u_2 - s.b) / s.a, u_2), q,
    (c_1, u_2))\}$
    \ENDIF
    \ENDIF
    \RETURN $S$
  \end{algorithmic}
  \label{alg:filter-segment}
\end{algorithm}

Now in order to filter a whole LB set using an UB set, we filter each segment
in the LB set with each point in the UB set.
If the remaining LB set is empty then the node can be fathomed, otherwise
branching is required.

\subsection{Branching procedure}
\label{sec:branching}
When dealing with bi-objective bound sets, a major challenge
is that the LB
set at a given node of the branch-and-bound tree may be partially dominated but
not completely, therefore the node cannot be fathomed. 
This is illustrated in
Figure~\ref{fig:splitLB}.

\begin{figure}
  \begin{center}
    \subfigure{
      \includegraphics[scale=.35]{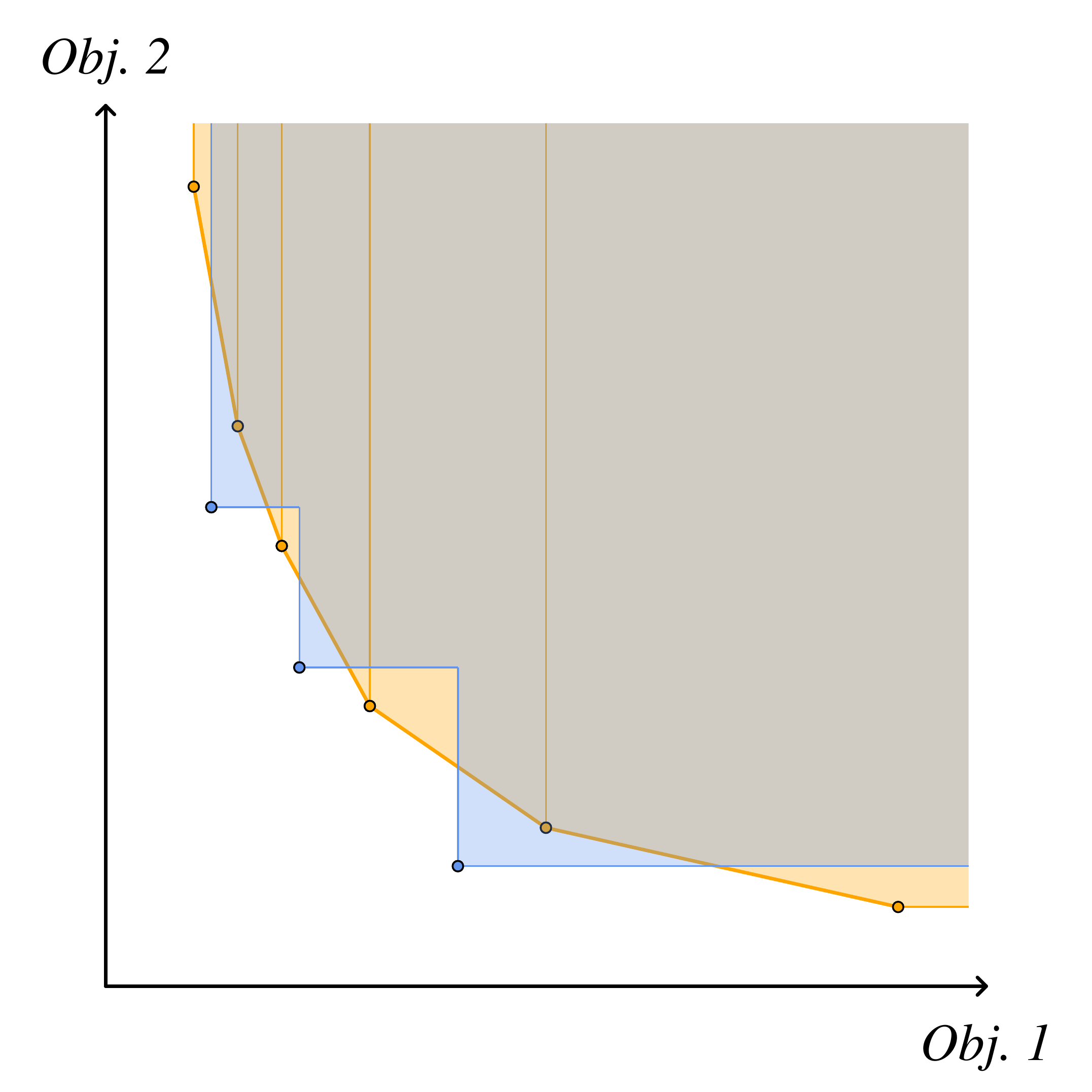}
    }
    \subfigure{
      \includegraphics[scale=.35]{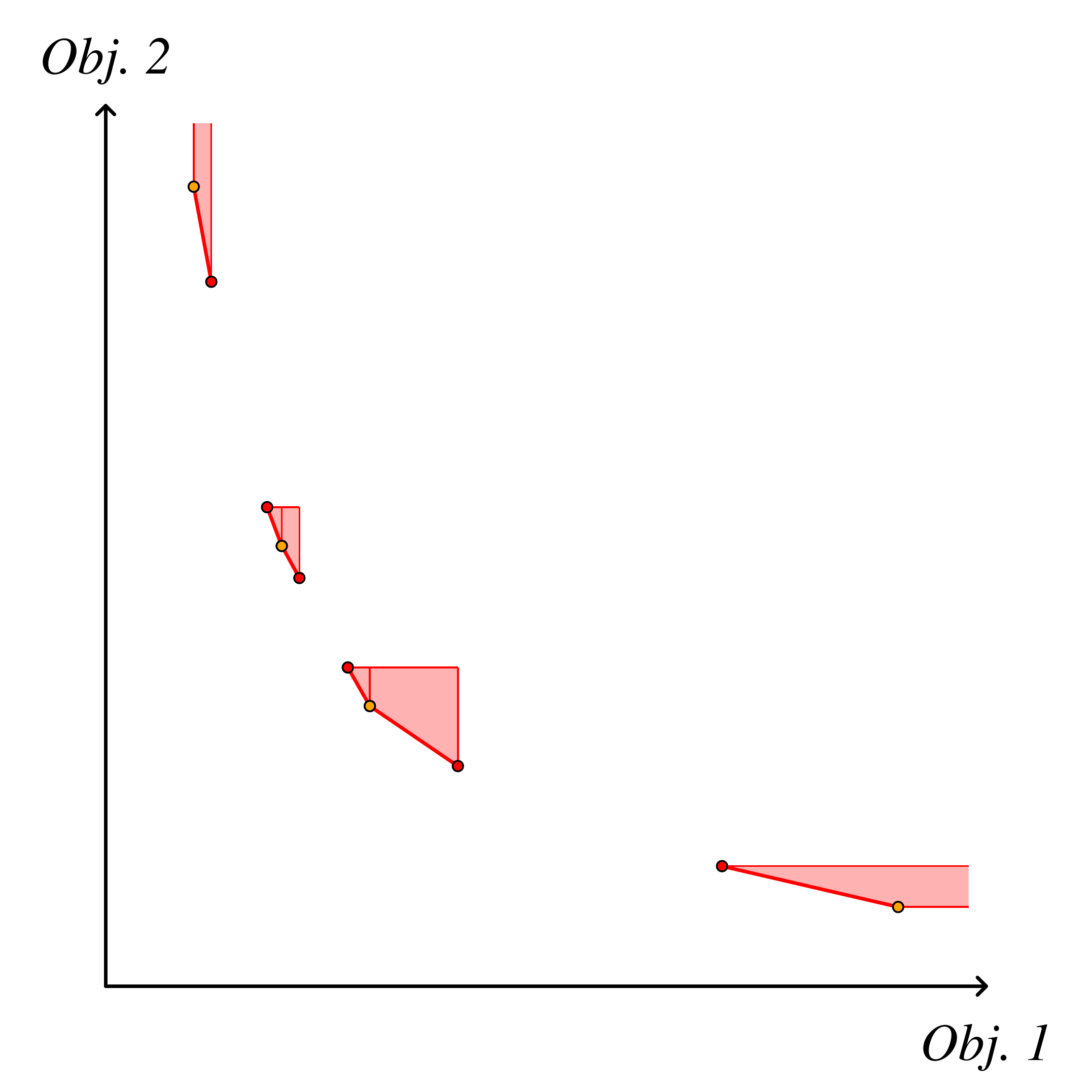}
    }
  \end{center}
  \caption{LB set is partially dominated by UB set, the node cannot be
    fathomed.}
  \label{fig:splitLB}
\end{figure}

In order to disregard those portions of the objective space which are dominated
by the UB set, we use the following branching rule which we name \emph{objective space branching}:
each contiguous
subset of the filtered LB set gives rise to a new branch and the
objective space between those parts is simply ignored. To generate the branch
for a given subset, we consider its local nadir point $c = (c_1, c_2)$ and add
two cutting 
planes
so that both objectives are better than the value at this local nadir point:
$f_1(x) \leq c_1$ and $f_2(x) \leq c_2$. In practice, this is achieved by
setting the objective bounds of the corresponding child nodes using
$c$.

Branching on objective space has no impact on decision
space, i.e. computing lower bounds for these child nodes will generate the same
LB set, albeit in several pieces, one per branch in objective space. Therefore,
objective space branching (OSB) never happens alone; branching on decision space is systematically also
performed at the same time. The end result is that each node in the search tree
contains at least one new branching decision on decision space, and optionally
a branching decision on objective space.

In
general, decision-space branching is problem specific. However, we can still
make the general remark that there is no reason 
to systematically use the same variable for decision-space branching in every
OSB child node: each OSB child has its own separate LB set, and
there is no reason to branch on the same decision variables when considering
different disjoint LB sets. For this reason, OSB is always applied first in
order to reduce the search space; for each OSB child and associated LB set,
decision-space branching is conducted independently.
The end result might still
be that each OSB child branches on the same variable, depending on the branching
rule.

Algorithm~\ref{alg:branch} gives the $branch$ function of our branch-and-bound
framework. It takes a node of the branch-and-bound tree and an LB set as input
and returns a set of nodes. A node is 
represented as a pair $(c, \Delta)$ where $c$ represents the objective
bounds (local nadir point) for this node and $\Delta$ is the set of branching
decisions for this node.

\begin{algorithm}[t]
  \caption{$branch(c, \Delta, LB)$}
  \begin{algorithmic}[1]
     \STATE $N \gets \emptyset$
    \FORALL {$S \in disjointPortions(LB)$}
    \STATE $c' = (\max\limits_{s \in S}(s.c_1), \max\limits_{s \in
      S}(s.c_2))$
    \STATE $D \gets decisionSpaceBranches(S)$ \label{line:psb}
    \FORALL {$ d \in D$}
    \STATE $N \gets N \cup (c', \Delta \cup \{d\})$
    \ENDFOR
    \ENDFOR
    \RETURN $N$
  \end{algorithmic}
  \label{alg:branch}
\end{algorithm}
In the example of Figure~\ref{fig:splitLB}, function $branch$ considers four
disjoint regions, each of them generating their own set of branching decisions.

We note here that even if an IP is used to compute the LB set, any
given variable may take different values over a same LB set; therefore
branching on decision space is still necessary.

\section{Improvements based on the integrality of objective
  functions}
\label{sec:integralityImprovements}
When solving integer programs, in many cases, objective values of feasible
integer solutions only take integer values.
For pure integer programs,
it is in practice almost always possible to use integer coefficients. The
reasons 
include the fact that LP solvers have precision limitations, and that
time-efficient floating-point numbers representations also have precision
limitations. So rounding floating-point numbers is almost always inevitable,
and if numbers are rounded to $d$ decimals then they may as well be multiplied
by $10^d$ and considered integers. We note that this property is
exploited by other methods as well. For instance, the $\epsilon$-constraint
framework uses a known $\epsilon$ value which in our case would be~1. The
balanced box method also relies on similar $\epsilon$ values.
In some cases, it is even possible to use values higher than 1, as long as
these values are valid: for instance if every coefficient for a given objective
function is integer, then the greatest common divisor of these coefficients can
be used as a valid value.
For the sake of simplicity and without loss of generality, we consider in the
following that this valid value is 1. We now explore possibilities to enhance
our bi-objective branch-and-bound by exploiting the fact that objective values
of feasible solutions are always integer.

Any given LB segment covers a part of the objective space, but also includes
subsets not containing any point with integer coordinates. Such subsets
can be disregarded during the search. 
 This is
illustrated in
Figure~\ref{fig:tighten-segment}, where dashed lines represent integer values
for each objective.
\begin{figure}
  \begin{center}
    \subfigure{
      \begin{tikzpicture}
        \node[anchor=south west,inner sep=0] (image) at (0,0) {\includegraphics[scale=.35]{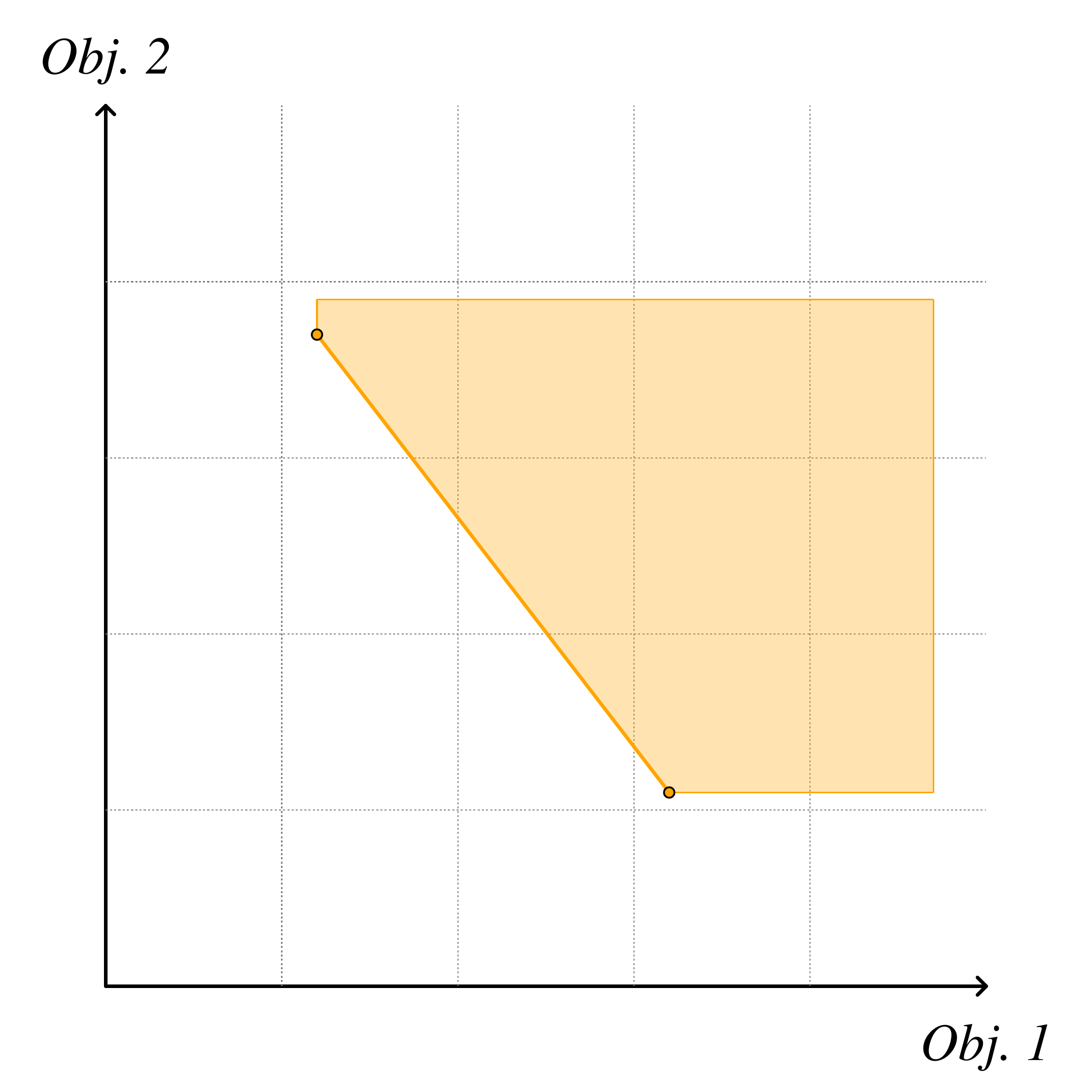}};
        \begin{scope}[x={(image.south east)},y={(image.north west)}]
          \large
          \node[anchor=north east] at (0.3,0.69) {$s.p$};
          \node[anchor=north east] at (0.62,0.28) {$s.q$};
          \node[anchor=south west] at (0.84,0.72) {$s.c$};
        \end{scope}
      \end{tikzpicture}
    }
    \subfigure{
      \begin{tikzpicture}
        \node[anchor=south west,inner sep=0] (image) at (0,0) {\includegraphics[scale=.35]{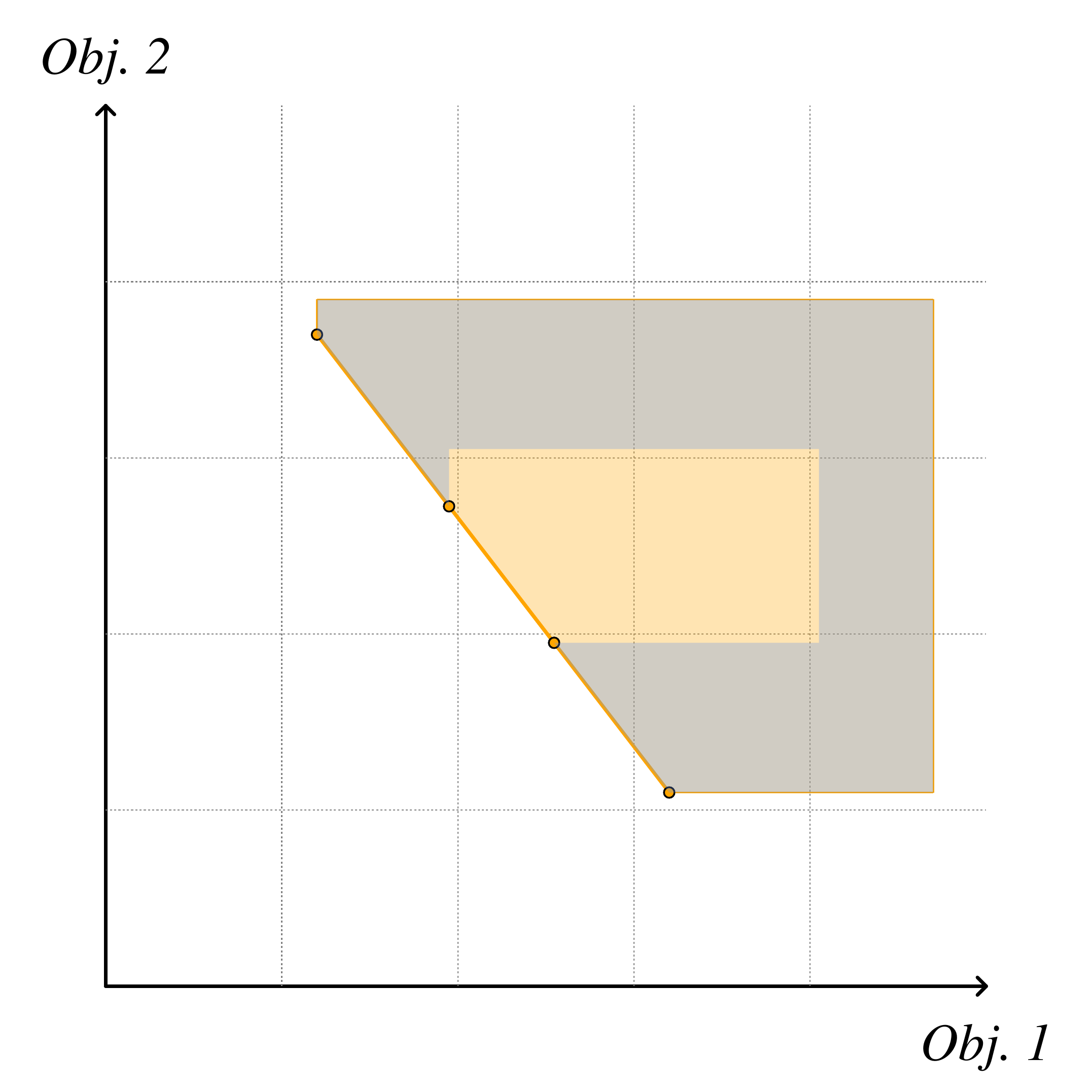}};
        \begin{scope}[x={(image.south east)},y={(image.north west)}]
          \large
          \node[anchor=north east] at (0.3,0.69) {$s.p$};
          \node[anchor=north east] at (0.62,0.28) {$s.q$};
          \node[anchor=south west] at (0.84,0.72) {$s.c$};
           \node[anchor=north east] at (0.43,0.56) {$s'.p$};
          \node[anchor=north east] at (0.52,0.44) {$s'.q$};
          \node[anchor=south west] at (0.73,0.57) {$s'.c$};
        \end{scope}
      \end{tikzpicture}
    }
  \end{center}
  \caption{The space covered by a LB segment contains contiguous regions which
    do not contain any point with integer coordinates; these parts are
    irrelevant to the search. Tightened segment $s'$ contains all the integer
    solutions contained by original segment $s$.}
  \label{fig:tighten-segment}
\end{figure}
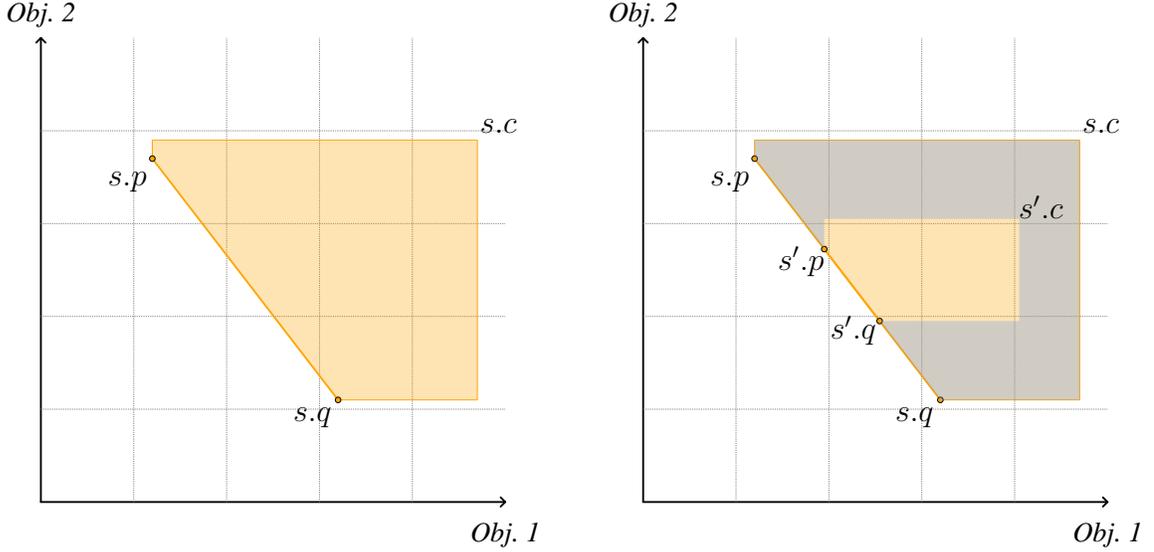
This general idea can be exploited in several ways in order to speed up the
search.

\subsection{Segment tightening}

First, any LB segment which does not cover any integer vector can be discarded. This can be
tested in $\mathcal{O}(1)$.

We first investigate the case where $s.p \neq s.q$, then $s.a$ and $s.b$
are defined:
\begin{proposition}
Segment $s$ with $s.p \neq s.q$ covers integer vectors 
iff $\lfloor s.c_2 \rfloor \geq s.a \lfloor s.c_1 \rfloor
+ s.b \land \lfloor s.c_2 \rfloor \geq s.q_2 \land \lfloor
s.c_1 \rfloor \geq  s.p_1$.
\end{proposition}
\begin{proof}
A segment $s$ is a polygon defined by the points $s.p$, $s.q$, $(s.c_1,s.q_2)$, $s.c$, $(s.p_1,s.c_2)$. The coordinates of any integer point $z = (z_1,z_2)$ within this polygon have to satisfy $s.p_1\leq z_1 \leq s.c_1$ and $s.q_2 \leq z_2 \leq s.c_2$ and $z_2 \geq s.a \cdot z_1 + s.b$. The integer point closest to the local nadir point $s.c$ which has the potential to lie within the polygon has coordinates $(\lfloor s.c_1 \rfloor, \lfloor s.c_2 \rfloor)$. Any other integer point $z = (z_1,z_2)$ contained within the region of the segment must have coordinates  $z_1 \leq \lfloor s.c_1 \rfloor$ and $z_2 \leq \lfloor s.c_2 \rfloor$. Therefore, if $(\lfloor s.c_1 \rfloor, \lfloor s.c_2 \rfloor)$ is contained within $s$, $s$ covers at least one integer point. $\square$
\end{proof}

The special case of segments defined by two points only (see
Definition~\ref{def:onepointsegment}) is even simpler, as the space covered by such a
segment is a rectangle. Consider the segment defined by point $p = (p_1, p_2)$ and local
nadir point $c = (c_1, c_2)$, then this segment contains integer points
iff $\lfloor c_1 \rfloor \geq z_1 \land \lfloor c_2 \rfloor \geq z_2$.

This is tested right after filtering:
only the segments in the LB set that cover an integer point are kept.
This enhancement is called \emph{segment tightening}.
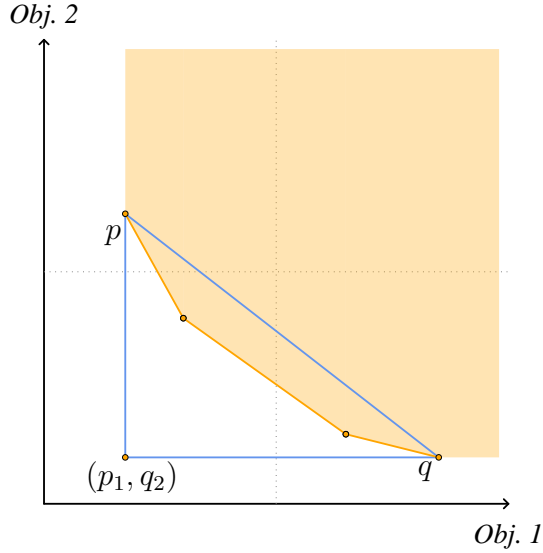
\begin{figure}
  \begin{center}
      \begin{tikzpicture}
        \node[anchor=south west,inner sep=0] (image) at (0,0) {\includegraphics[scale=.35]{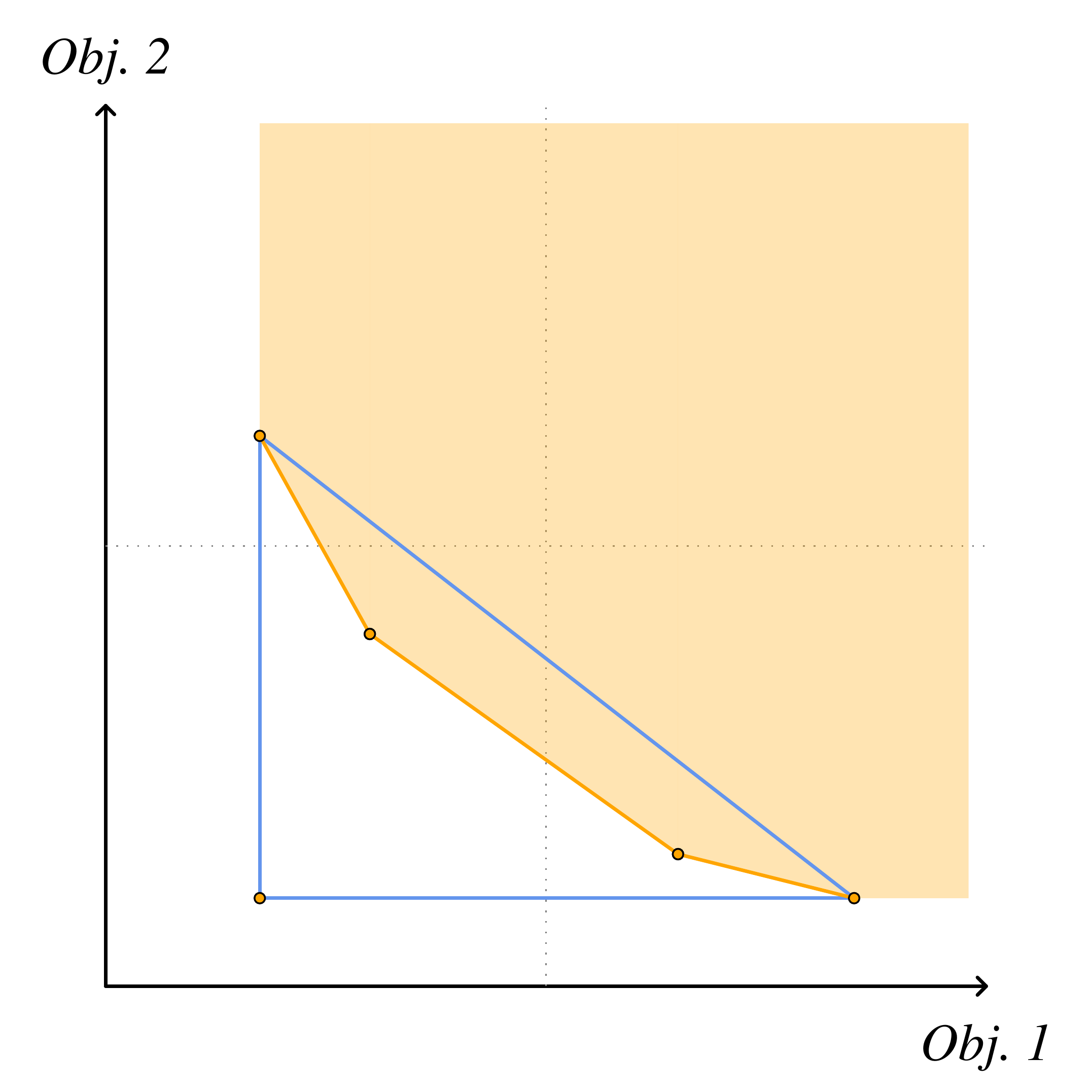}};
        \begin{scope}[x={(image.south east)},y={(image.north west)}]
          \large
          \node[anchor=north east] at (0.25,0.6) {$p$};
          \node[anchor=north east] at (0.79,0.19) {$q$};
          \node[anchor=north ] at (0.25,0.19) {$(p_1, q_2)$};
        \end{scope}
      \end{tikzpicture}
  \end{center}
  \caption{Segment between points $p = (p_1, p_2)$ and $q =(q_1,q_2)$ defines a valid LB
    segment because
    there is no vector of integer coordinates in the triangle defined by $p$,
    $q$ and $(p_1, q_2)$.}
  \label{fig:liftHull}
\end{figure}

\subsection{Lower bound lifting}
\label{sec:LBlifting}
Second, we can use the fact that feasible integer solutions have integer
objective values in order to speed up our LB set computation procedure.
We call this procedure $bound$ (it is given in Algorithm~\ref{alg:aneja}).

\begin{algorithm}
  \caption{$bound( node, UB )$}
   \begin{algorithmic}[1]
    \STATE $C \gets \emptyset$, $E \leftarrow \emptyset$
    \STATE $\bar{x}^1 \gets solve(lexmin(f_1, f_2), node, UB)$
    \IF {$\bar{x}^1= infeasible$}
    \RETURN $E$
    \ENDIF
    \STATE $\bar{x}^2 \gets solve(lexmin(f_2, f_1), node, UB )$
    \STATE $p \gets (f_1(\bar{x}^1),f_2(\bar{x}^1))$, $q \gets (f_1(\bar{x}^2),f_2(\bar{x}^2))$
    \STATE $push( C, (p,q))$
    \WHILE {$C \neq \emptyset$}
    \STATE $(p, q) \gets pop(C)$ \label{algoline:aneja:popC}
    \STATE $s \leftarrow Segment(p, q, (\infty,\infty))$
    \IF {$\lceil q_2 \rceil \geq s.b + s.a \lceil p_1\rceil$} 
       \STATE $E \gets E \cup \{(p, q)\}$ 
    \ELSE	
    \STATE $w_2 \gets q_1 - p_1$
    \STATE $w_1 \gets p_2 - q_2$
    \STATE $\bar{x} \gets  solve(w_1 f_1 + w_2 f_2, node, UB)$
     \IF {$w_1 f_1(\bar{x}) + w_2 f_2(\bar{x}) < w_1 p_1 + w_2 p_2 $}
    \STATE $u \gets (f_1(\bar{x}),f_2(\bar{x}))$
    \STATE $push( C, (p,  u))$
    \STATE $push( C, (u, q))$
    \ELSE
    \STATE $E \gets E \cup \{(p, q)\}$ 
    \ENDIF
    \ENDIF
    \ENDWHILE
    \RETURN $E$
  \end{algorithmic}
  \label{alg:aneja}
\end{algorithm}
As in the original algorithm of~\cite{Aneja:1979}, procedure $bound$
successively computes supported points, starting with the two
points
at the two ends of the efficient frontier, using lexicographic minimum
objective functions $lexmin(f_1, f_2)$ and $lexmin(f_2, f_1)$. This is similar
to what is done by~\cite{Boland:2015} in the context of the balanced box method.

For any two given points, solving a certain
weighted-sum problem determines whether there exists another supported solution
``between'' these two.
However, it can happen that these two points already describe a valid LB
segment (given an appropriate nadir point), regardless of any other extreme
supported point between them.
If there is no point with 
integer coordinates in the triangle defined by $p = (p_1, p_2)$, $q = (q_1,
q_2)$ and $(p_1, q_2)$, 
then $p$ and $q$ already define a valid LB segment and there is no 
need to investigate these two points further.
We note that point $(p_1, q_2)$ can also be disregarded because it 
dominates both $p$ and $q$, although it is already established
that they are both non-dominated.
Moreover, if the line connecting $p$ and $q$ is not a part of the convex hull
boundary then 
it defines a tighter lower bound than the one induced by the actual convex hull
boundary between $p$ and $q$.
This situation is illustrated in Figure~\ref{fig:liftHull}, where
$p$ and $q$ define a valid lower bound segment (given an
appropriate local nadir point).

As previously, testing for the existence of an integer vector in the
given triangle can be achieved in $\mathcal{O}(1)$ and sometimes allows to
skip some stages of the algorithm while providing a tighter LB set. 
In the
notation of Algorithm~\ref{alg:aneja}, this means that after retrieving
$(p, q)$ from $C$
(line~\ref{algoline:aneja:popC}, Algorithm~\ref{alg:aneja}), we check if
the triangle derived from $(p,q)$ does not contain any
integer point. 
\begin{proposition}
If the point $(\lceil p_1\rceil, \lceil q_2 \rceil)$ is above the line
connecting $p$ and $q$, i.e. $\lceil q_2 \rceil \geq s.a \lceil p_1\rceil +
s.b$, then
the triangle derived from $(p,q)$ does not contain any
integer point.
\end{proposition}
\begin{proof}
Since $p$ and $q$ are extreme points of the boundary of the convex hull, an additional LB point may only be found in the right triangle defined by the points $p$, $(p_1, q_2)$, $q$, with $p_1 \leq q_1$ and $q_2 \leq p_2$. Let us call $(p_1, q_2)$ the local ideal point.
Since $\lceil p_1 \rceil \geq p_1$ and $\lceil q_2 \rceil \geq  q_2$, $(\lceil p_1\rceil, \lceil q_2 \rceil)$ must either lie within the triangle or it is above the line connecting $p$ and $q$. Since any other integer point $z = (z_1,z_2)$ must have coordinates satisfying $z_1 \geq \lceil p_1\rceil$ and $z_2 \geq \lceil q_2 \rceil$, if $(\lceil p_1\rceil, \lceil q_2 \rceil)$ lies above the line connecting $p$ and $q$, no other integer point within this triangle exists.$\square$
\end{proof}
In this case, the segment derived from $(p, q)$ is appended to $E$ in Algorithm~\ref{alg:aneja} and lines 13--21 are skipped.
Otherwise a weighted sum problem ($w_1 f_1 + w_2 f_2$) is solved, where
$w_1$ and $w_2$ are determined by the algorithm. In the case where the
resulting point $u$ lies below the line connecting $p$ and $q$, $(p, u)$
and $(u,q)$ are stored in $C$ for further processing, as in the algorithm by~\cite{Aneja:1979}.
This enhancement is called \emph{lower bound lifting}.

\subsection{Integer dominance}

Third, it is also possible to exploit the
integrality of objective values for feasible integer solutions when branching
on objective space. If we are in
the situation of branching on objective space, it is typically because part of
the LB set has been dominated by the UB set. In that case, we already know that
the given local nadir point of any LB segment is in fact also dominated by the
UB set (this is actually how this local nadir point is calculated, see
Algorithm~\ref{alg:filter-segment}). So this local nadir point can be
tightened.
\begin{figure}
  \begin{center}
    \subfigure[]{
      \begin{tikzpicture}
        \node[anchor=south west,inner sep=0] (image) at (0,0) 
        { \includegraphics[scale=.35]{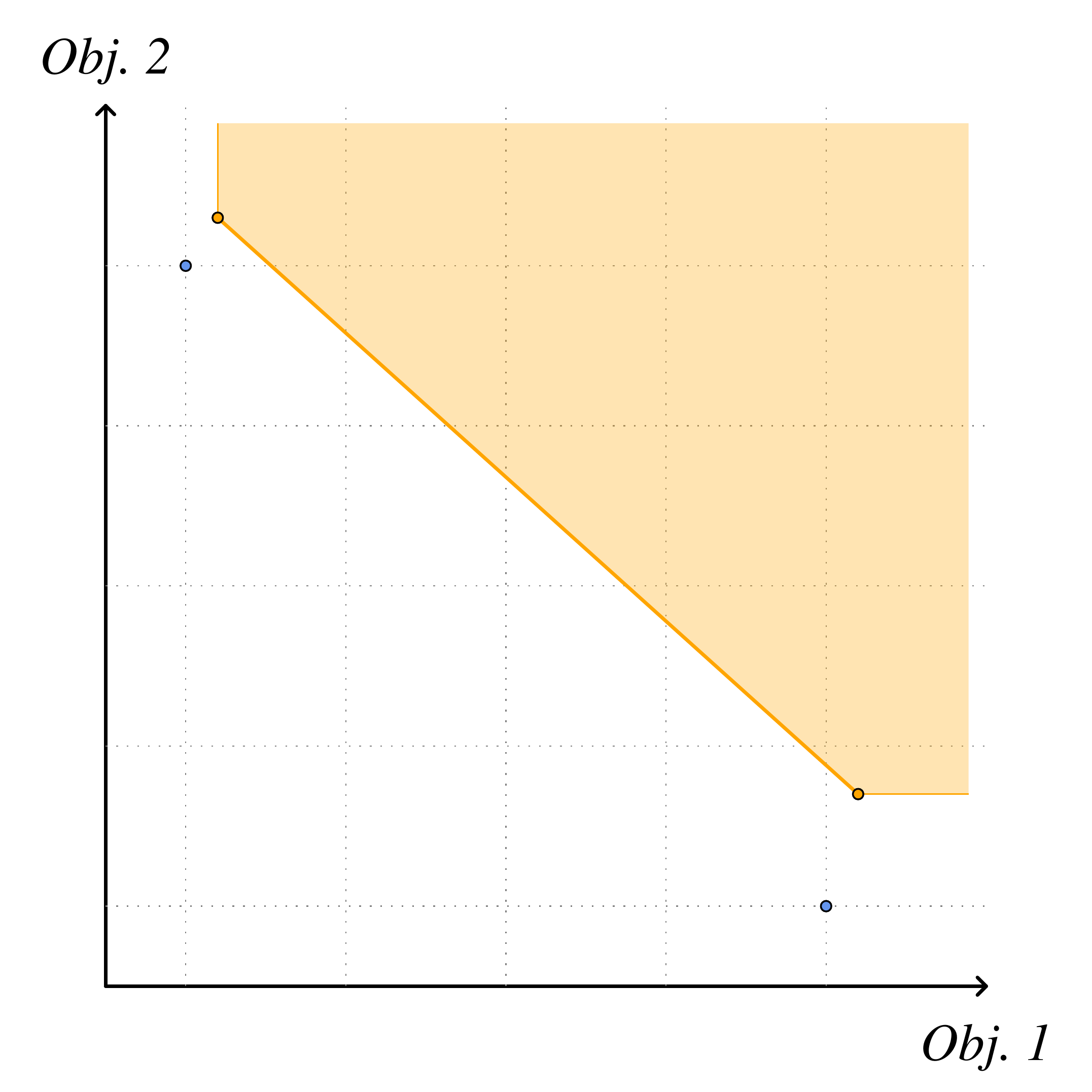}};
        \begin{scope}[x={(image.south east)},y={(image.north west)}]
          \large
          \node[anchor=north east] at (0.32,0.85) {$s.p$};
          \node[anchor=north east] at (0.9,0.33) {$s.q$};
           \node[anchor=north east] at (0.18,0.75) {$u$};
          \node[anchor=north east] at (0.76,0.18) {$v$};
        \end{scope}
      \end{tikzpicture}
      \label{subfig:lift1}
    }
    \subfigure[]{
      \begin{tikzpicture}
        \node[anchor=south west,inner sep=0] (image) at (0,0) 
        { \includegraphics[scale=.35]{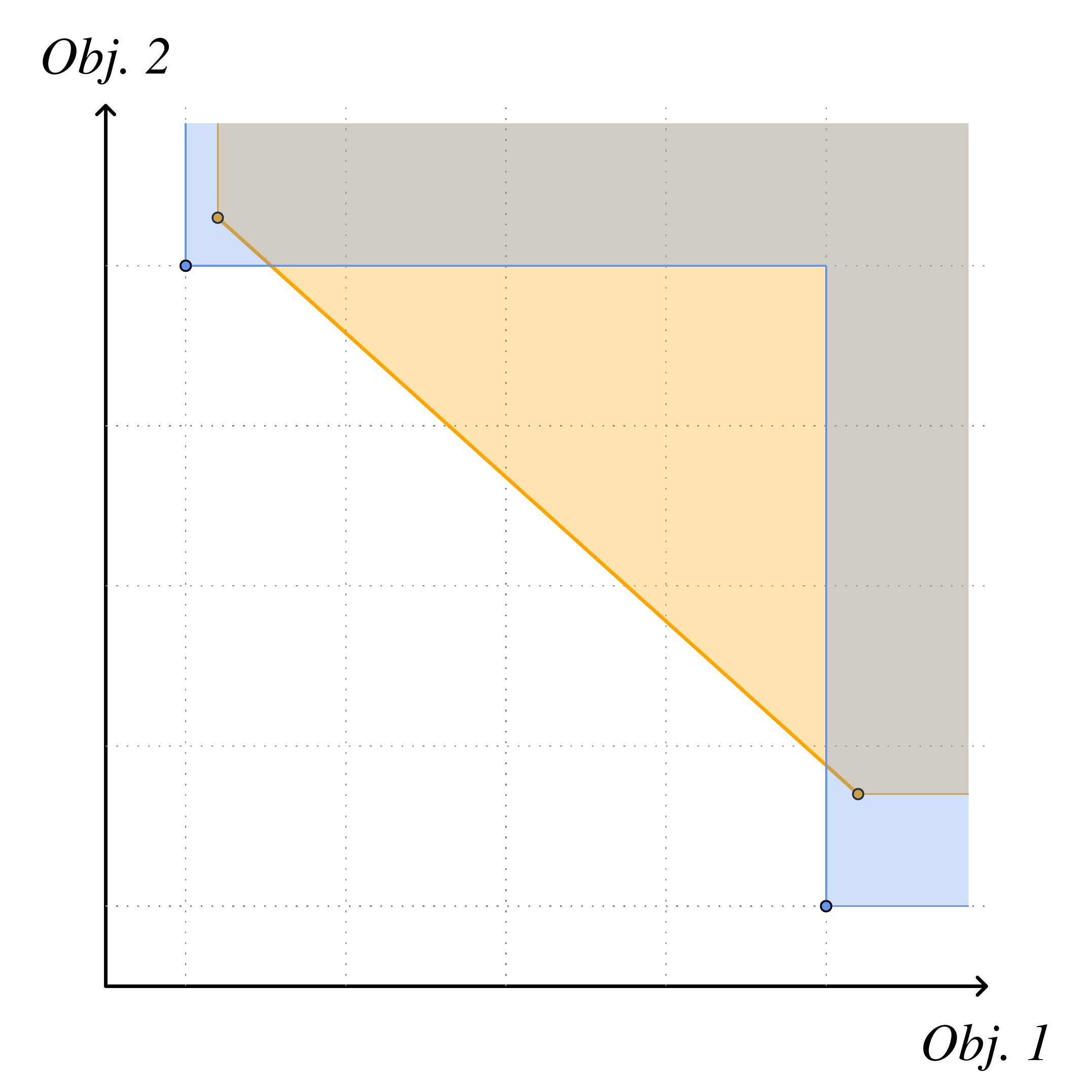}};
        \begin{scope}[x={(image.south east)},y={(image.north west)}]
          \large
          \node[anchor=north east] at (0.32,0.85) {$s.p$};
          \node[anchor=north east] at (0.9,0.33) {$s.q$};
         \node[anchor=north east] at (0.18,0.75) {$u$};
          \node[anchor=north east] at (0.76,0.18) {$v$};
        \end{scope}
      \end{tikzpicture}
      \label{subfig:lift2}
    }
    \subfigure[]{
       \begin{tikzpicture}
        \node[anchor=south west,inner sep=0] (image) at (0,0) 
        { \includegraphics[scale=.35]{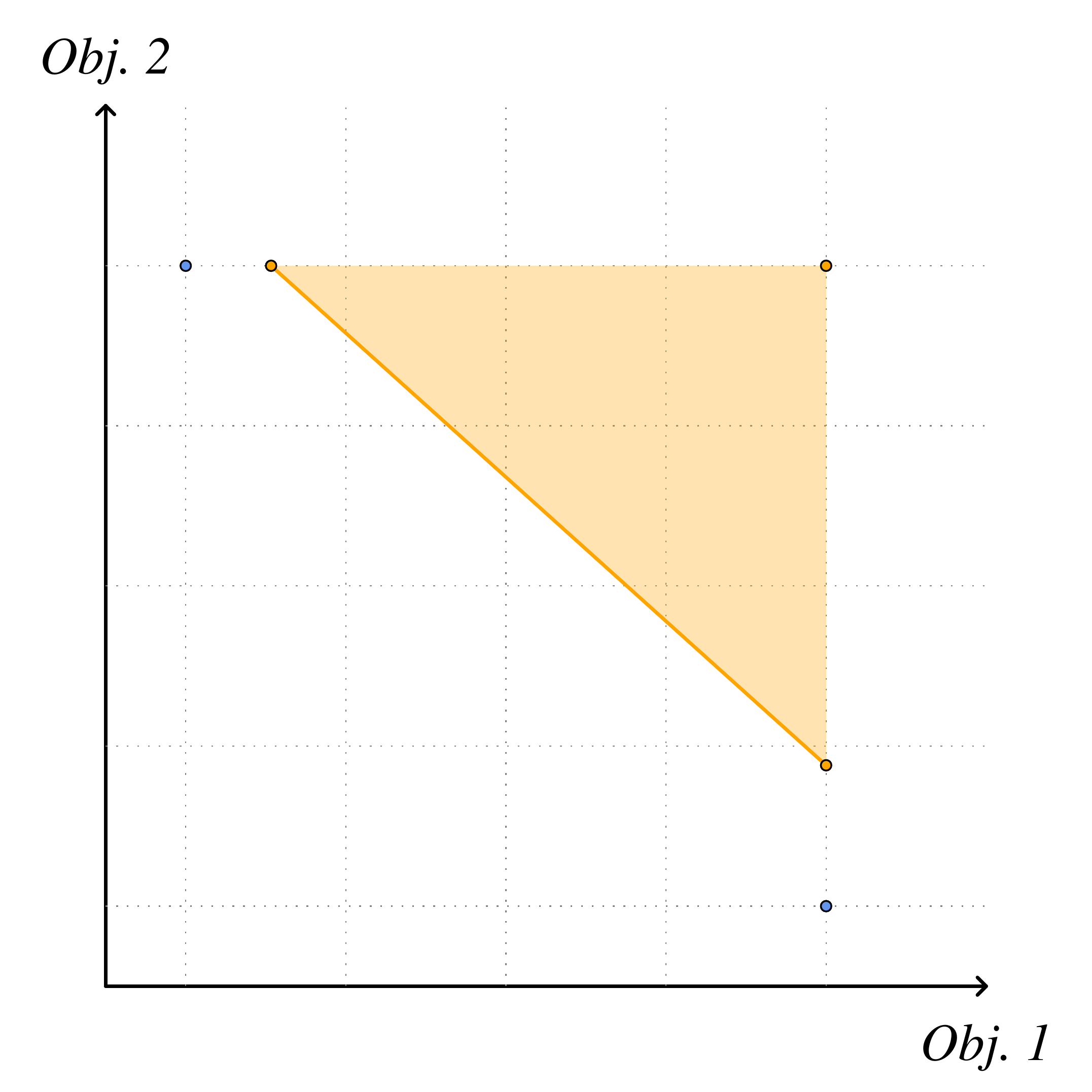}};
        \begin{scope}[x={(image.south east)},y={(image.north west)}]
          \large
          \node[anchor=south west] at (0.23,0.74) {$s'.p$};
          \node[anchor=north east] at (0.9,0.38) {$s'.q$};
         \node[anchor=south west] at (0.75,0.74) {$s'.c$};
           \node[anchor=north east] at (0.18,0.75) {$u$};
          \node[anchor=north east] at (0.76,0.18) {$v$};
        \end{scope}
      \end{tikzpicture}
      \label{subfig:lift3}
    }
    \subfigure[]{
      \begin{tikzpicture}
        \node[anchor=south west,inner sep=0] (image) at (0,0) 
        { \includegraphics[scale=.35]{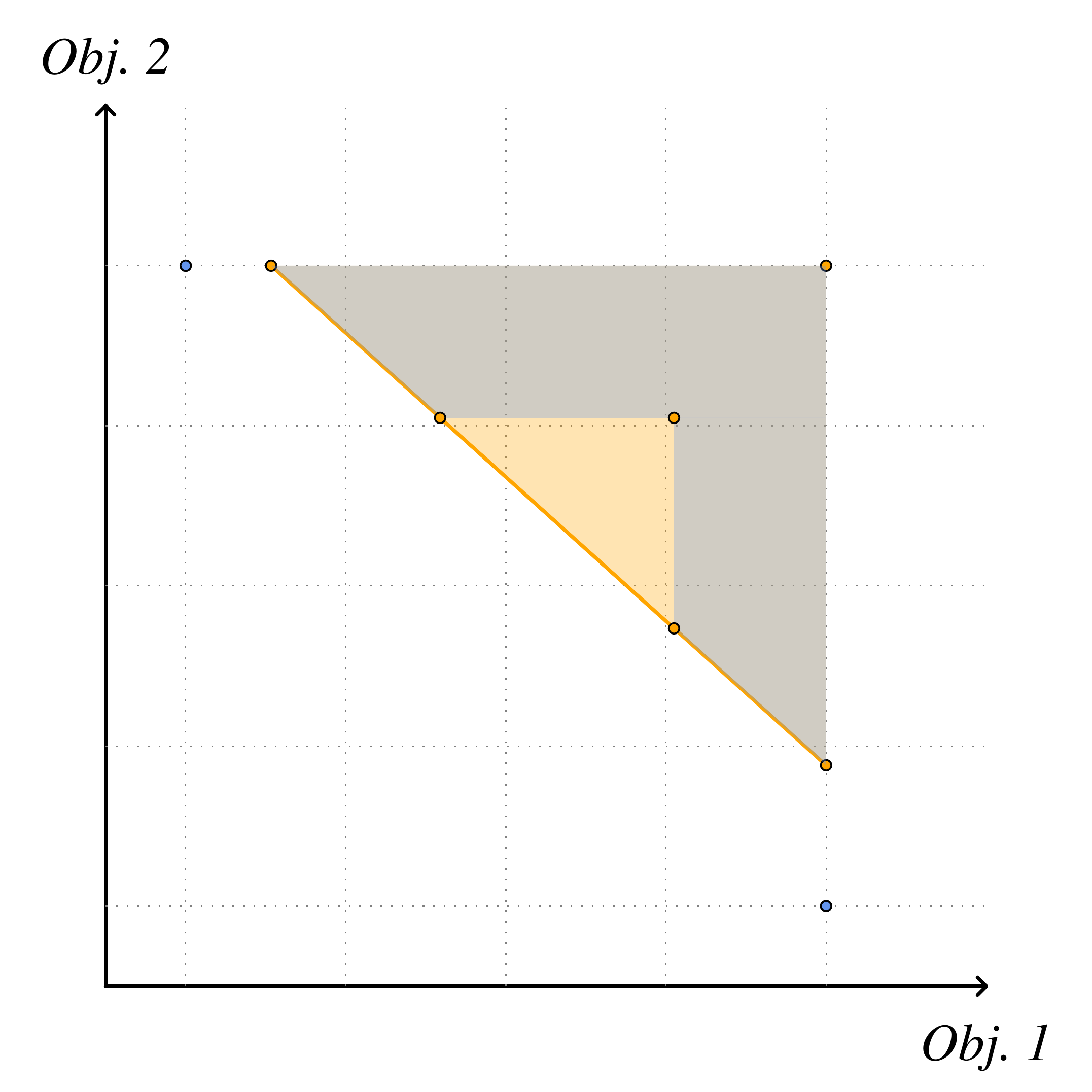}};
        \begin{scope}[x={(image.south east)},y={(image.north west)}]
          \large
             \node[anchor=south west] at (0.23,0.74) {$s'.p$};
          \node[anchor=north east] at (0.9,0.38) {$s'.q$};
         \node[anchor=south west] at (0.75,0.74) {$s'.c$};
             \node[anchor=south west] at (0.38,0.6) {$s''.p$};
          \node[anchor=north east] at (0.76,0.5) {$s''.q$};
         \node[anchor=south west] at (0.6,0.6) {$s''.c$};
           \node[anchor=north east] at (0.18,0.75) {$u$};
          \node[anchor=north east] at (0.76,0.18) {$v$};
        \end{scope}
      \end{tikzpicture}
      \label{subfig:lift4}
    }
  \end{center}
  \caption{Part of a LB segment $s$
  is either still dominated by the UB (represented by $u$ and $v$) or not
    containing any integer vector (darker shaded area in the fourth picture).}
  \label{fig:lift-objspace}
\end{figure}
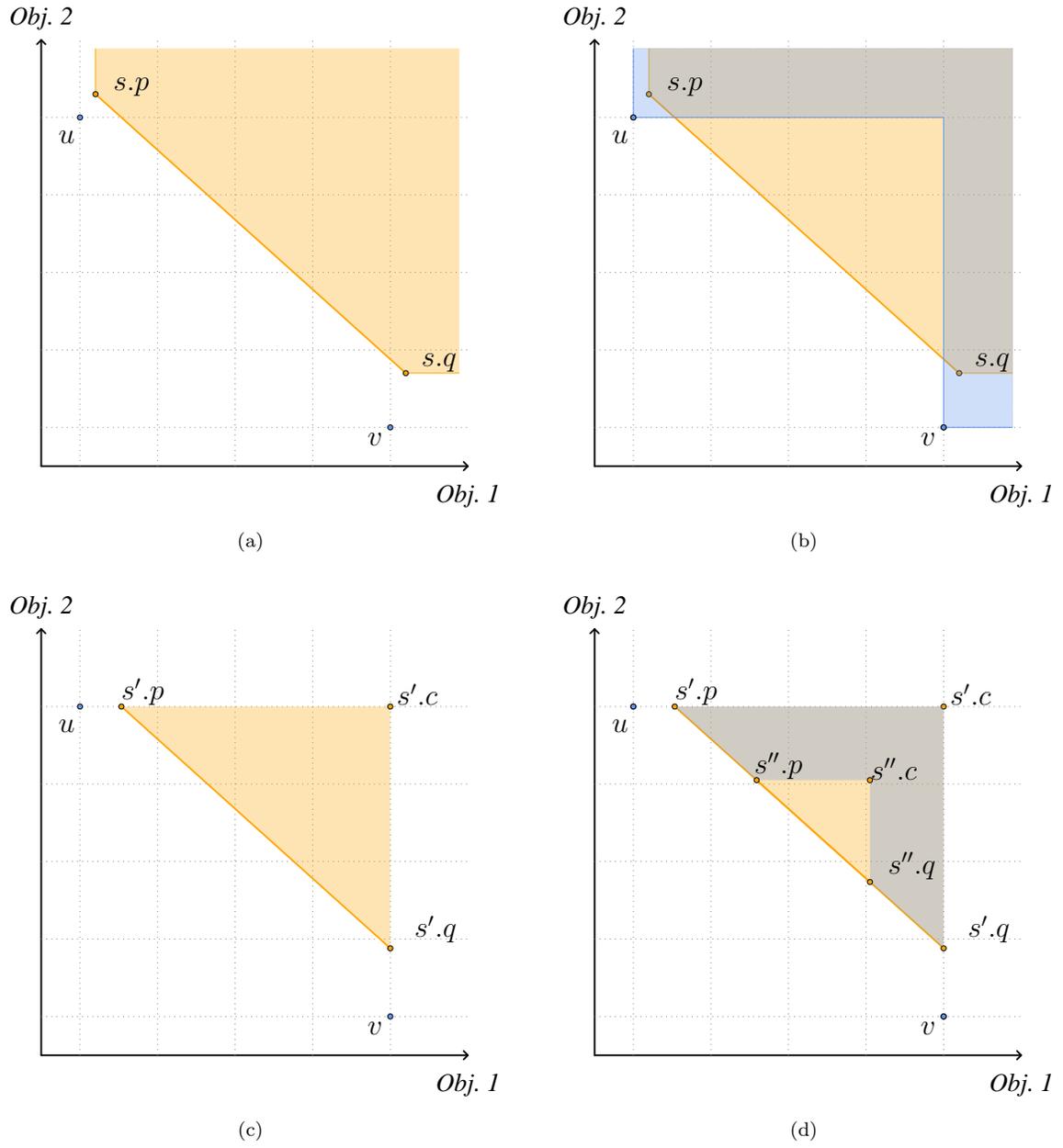
 Figure~\ref{fig:lift-objspace} illustrates this principle: after filtering the
 LB segment $s$ with points from the upper bound set, 
$u$ and $v$ (Figures~\ref{subfig:lift1}
and \ref{subfig:lift2}), we obtain a reduced segment $s'$. The new local nadir point $s'.c$
is actually dominated by both $u$ and $v$
(Figure~\ref{subfig:lift3}). Moreover, each point in the area which is in a
darker shade (Figure~\ref{subfig:lift4}) is either (i) dominated by $u$ and/or
$v$ or (ii) not integer. Therefore we can disregard the whole area. 
This can be achieved by subtracting any value $\rho$ such that $0 \leq \rho <
1$  from both coordinates of the local nadir point $s'.c$ after filtering,
thus producing an improved segment $s''$.
In order to implement this idea in our algorithm, we simply subtract $\rho$
from both coordinates of each upper bound point $u$ used in
Algorithm~\ref{alg:filter-segment}.
This enhancement is called \emph{integer dominance}. It is in fact similar to
using an $\epsilon$ value of 1 in the $\epsilon$-constraint framework, or in
the balanced box method.

\section{Computational study}\label{sec:exp}

We investigate the efficiency of the proposed algorithm and enhancements
through experimentation. 

A generic version of BIOBAB is developed in Python 2.7 
and Gurobi
6.5.0 is used to solve LPs and MIPs. Algorithms are run on an Intel Xeon
E5-2650v2 CPU at 2.6 GHz, with a 4 GB memory limit and a two-hour CPU time
limit. Except when explicitly stated otherwise, BIOBAB uses the linear
relaxation of the original IP to compute lower bound sets. In terms of tree
exploration, BIOBAB always performs a breadth-first search.
Like in~\cite{Boland:2015}, we use \emph{performance profiles} to compare the
performances of different algorithms (\cite{Dolan2002}). The
performance of an algorithm for a certain instance is the ratio of the CPU time
required by that algorithm for that instance to the best CPU time for that
instance. The
$x$-axis represents performance, while the $y$-axis represents the
fraction of instances solved at a certain performance or better.
If an algorithm does not terminate (i.e. find the Pareto set) for a
certain instance, then it does not offer a performance for that instance. If no
algorithm terminates for a given instance, then this instance is discarded.

In the following, we first describe the considered problems and benchmark
instances for the first set of experiments, where we evaluate the different
components of the algorithm and the impact of the integrality-based
improvements. We then compare the results of our method with those of criterion space search algorithms and to an existing branch-and-bound method. Thereafter, we formally introduce the bi-objective team orienteering problem with time windows (BITOPTW), whose linear relaxation is solved by means of column generation, we explain the different components tailored to this problem, and we report the obtained results.

\subsection{Considered problems and benchmark instances}

For our first set of experiments, we apply BIOBAB to two different uncapacitated
bi-objective facility location problems. The first uncapacitated facility
location problem, which we call \emph{uncapacitated bi-objective facility
  location problem} (UBOFLP), considers the cost of opening facilities as first
objective and the coverage of the population as second objective. Population
from a certain node can only be covered by an open facility if the distance
between the two is below a threshold. For the UBOFLP, we use the data that was
also used to generate instances from the bi-objective stochastic covering tour
problem presented in~\cite{tricoire2012}. The following simplifications are
performed: (i) the transportation costs are ignored, (ii) the stochastic aspect
is ignored and (iii) the distance aspect is simplified so that a node is either
entirely covered or not covered at all by an open facility. In order to produce
instances with various levels of difficulty, we concatenate the original
instances: the $k^{th}$ UBOFLP instance consists of the data of the
first $k$ covering tour problem instances. In total there are 36 UBOFLP
instances, including 21 to 500 locations.
The second
uncapacitated facility location problem we consider is a simplification of the
single-source capacitated facility location problem, where the capacity
constraints are relaxed. The first objective is to minimize the cost of opening
facilities, while the second objective is to minimize assignment costs. We name
this problem SSUFLP from now on. We use the 120 
instances
publicly available from \cite{gadegaard2016}. For both UBOFLP and SSUFLP, decision space branching is
performed in a generic way. First, the average value of
each binary variable over the whole LB set (or rather subset,
see Section~\ref{sec:branching}) is computed. Then, the variable with
average value closest to 1 is selected for binary branching on that subset.
 We provide mathematical formulations for UBOFLP and SSUFLP in
 Appendices~\ref{sec:uboflp-model} and \ref{sec:ssuflp-model}. For each of
 these problems, we compute the greatest common divisor for each objective and
 use it as valid $\epsilon$ value in the various methods. Since all costs are
 integer, $\epsilon$ is always at least 1.

We provide information on the cardinality of efficient sets for UBOFLP and
SSUFLP instances in Figure~\ref{fig:nPoints}.
\begin{figure}
  \centering
  \includegraphics[scale=.6]{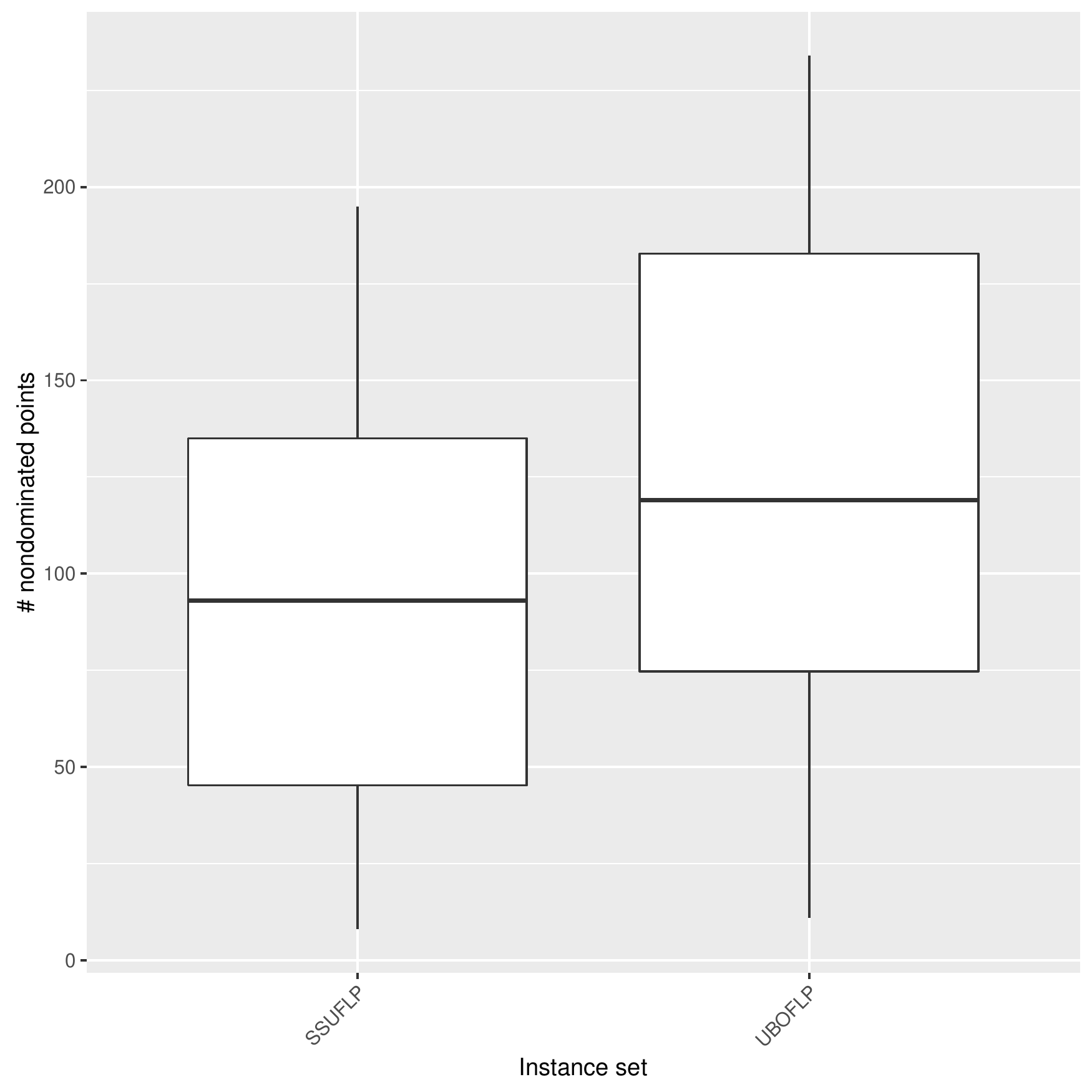}
  \caption[caption]{Number of non-dominated points for SSUFLP and UBOFLP instance
    sets.\\
    SSUFLP: 120 instances.\\
    UBOFLP: 36 instances.}
  \label{fig:nPoints}
\end{figure}
 
\subsection{Impact of integrality-based improvements}

In a first series of experiment, we assess the efficiency of the improvements
described in Section~\ref{sec:integralityImprovements}, using the UBOFLP and
SSUFLP data set. We run a version of
BIOBAB with no improvement at all, a version
with objective space branching only, a version with segment tightening only, a
version with lower bound lifting only, and a version with all three
improvements.
Performance profiles are
displayed in Figure~\ref{fig:biobab-improvements}.
\begin{figure}
  \subfigure[UBOFLP]{
    \includegraphics[scale=.4]{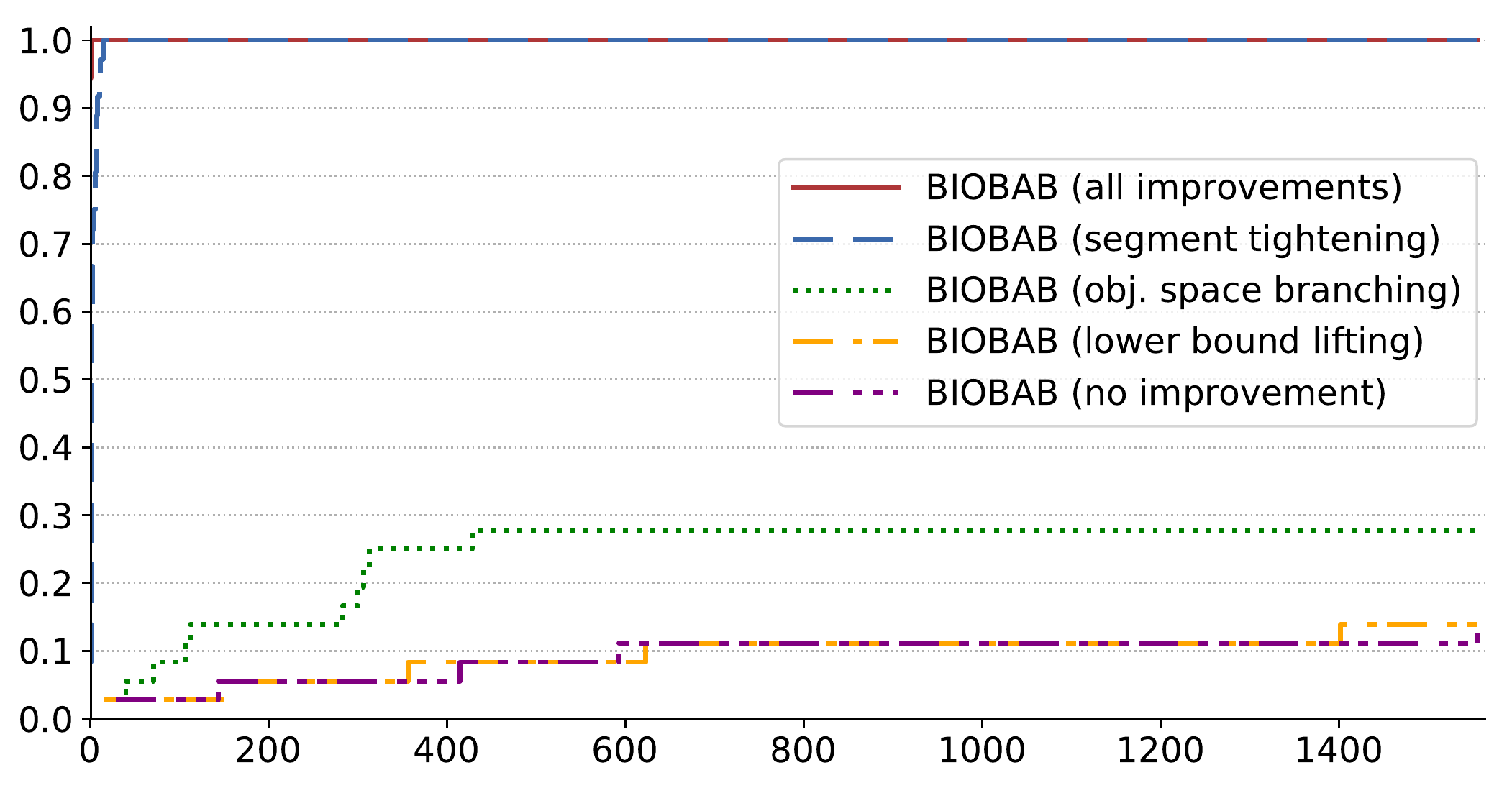}
    }
  \subfigure[SSUFLP]{
    \includegraphics[scale=.4]{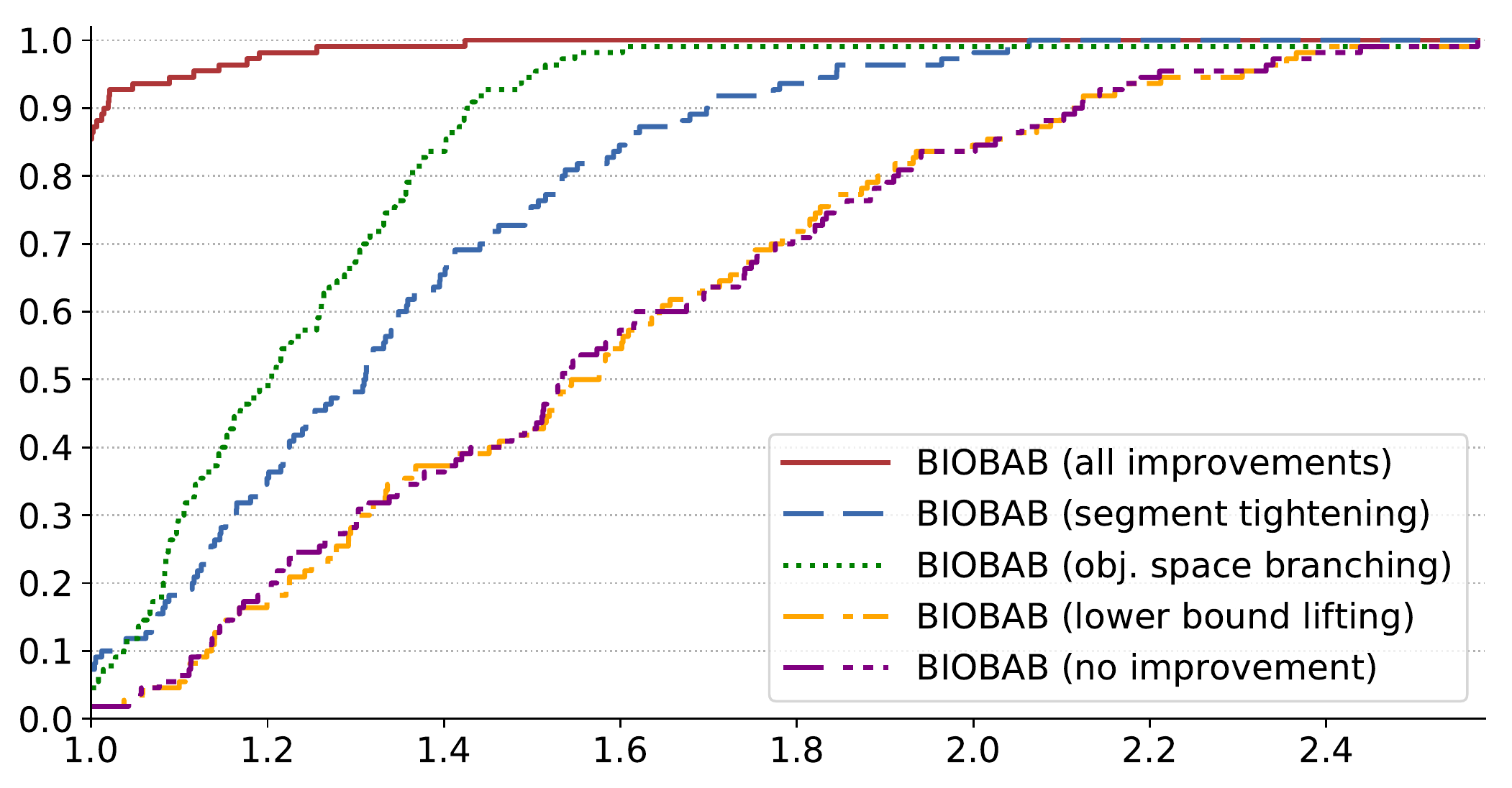}
    }
  \caption{BIOBAB: individual contribution of various improvements}
  \label{fig:biobab-improvements}
\end{figure}
On UBOFLP instances, full BIOBAB as well as segment tightening clearly
outperform the other variants. Whether full BIOBAB outperforms the version with
segment tightening only is unclear, because the whole chart
is scaled to also include the performance profiles of poorly performing
variants. For the same reason, it is also unclear whether lower bound lifting
brings a significant improvement over bare BIOBAB. On SSUFLP instances, the
differences are smaller but still clearly marked; full BIOBAB is still the
best, but this time it appears that objective-space branching is working better
on its own than segment tightening. Again, it is unclear whether lower bound
lifting brings any improvement.

We further investigate the individual contributions of the various
BIOBAB components, by comparing full BIOBAB to full BIOBAB minus individual
improvements. We first look at segment tightening, which appears to be crucial
in certain cases, as shown in
Figure~\ref{fig:biobab-improvements}. Figure~\ref{fig:biobab-full-vs-nst}
presents the performance of full BIOBAB with and without segment tightening.
\begin{figure}
  \subfigure[UBOFLP]{
    \includegraphics[scale=.4]{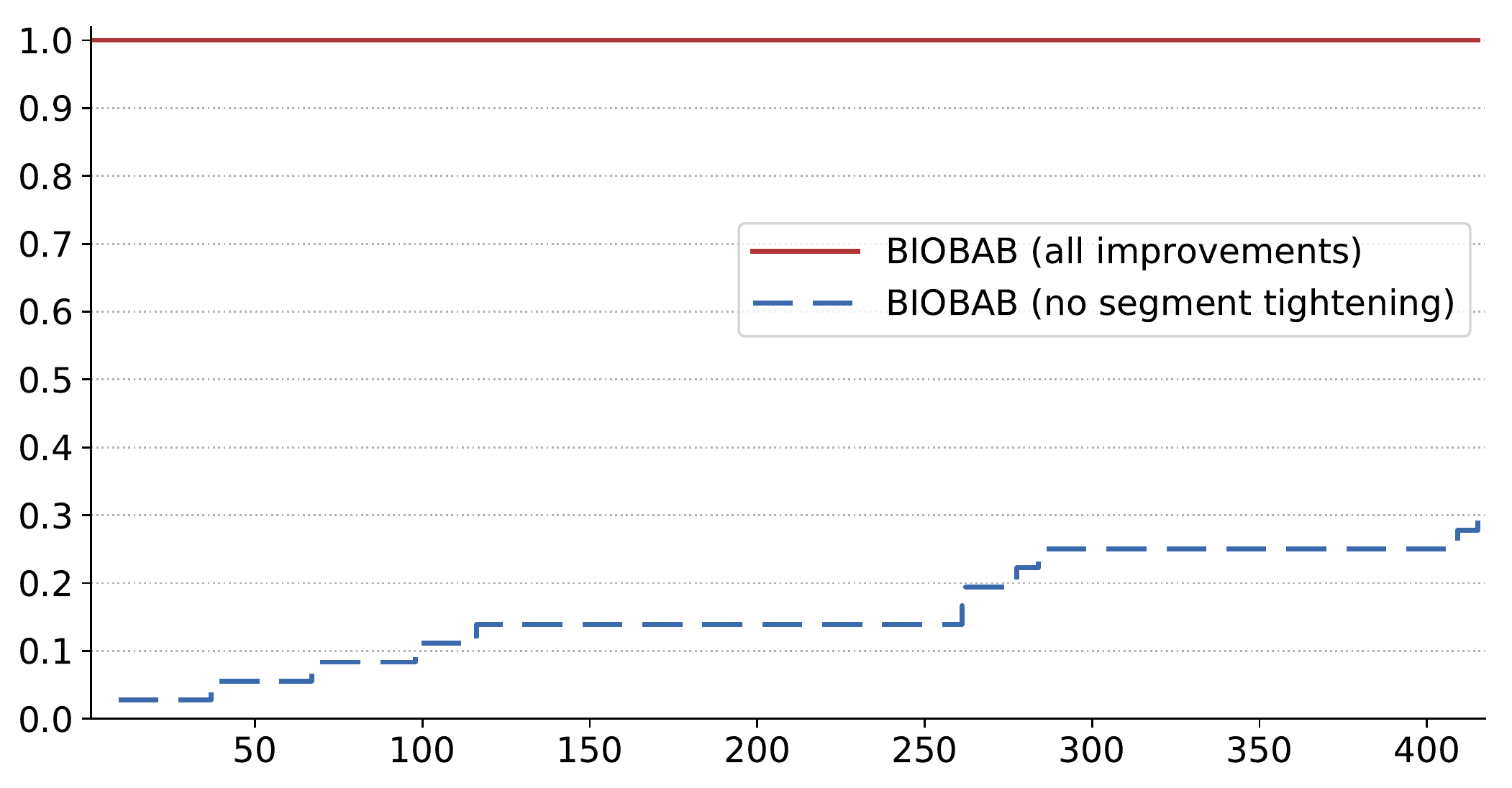}
    }
  \subfigure[SSUFLP]{
    \includegraphics[scale=.4]{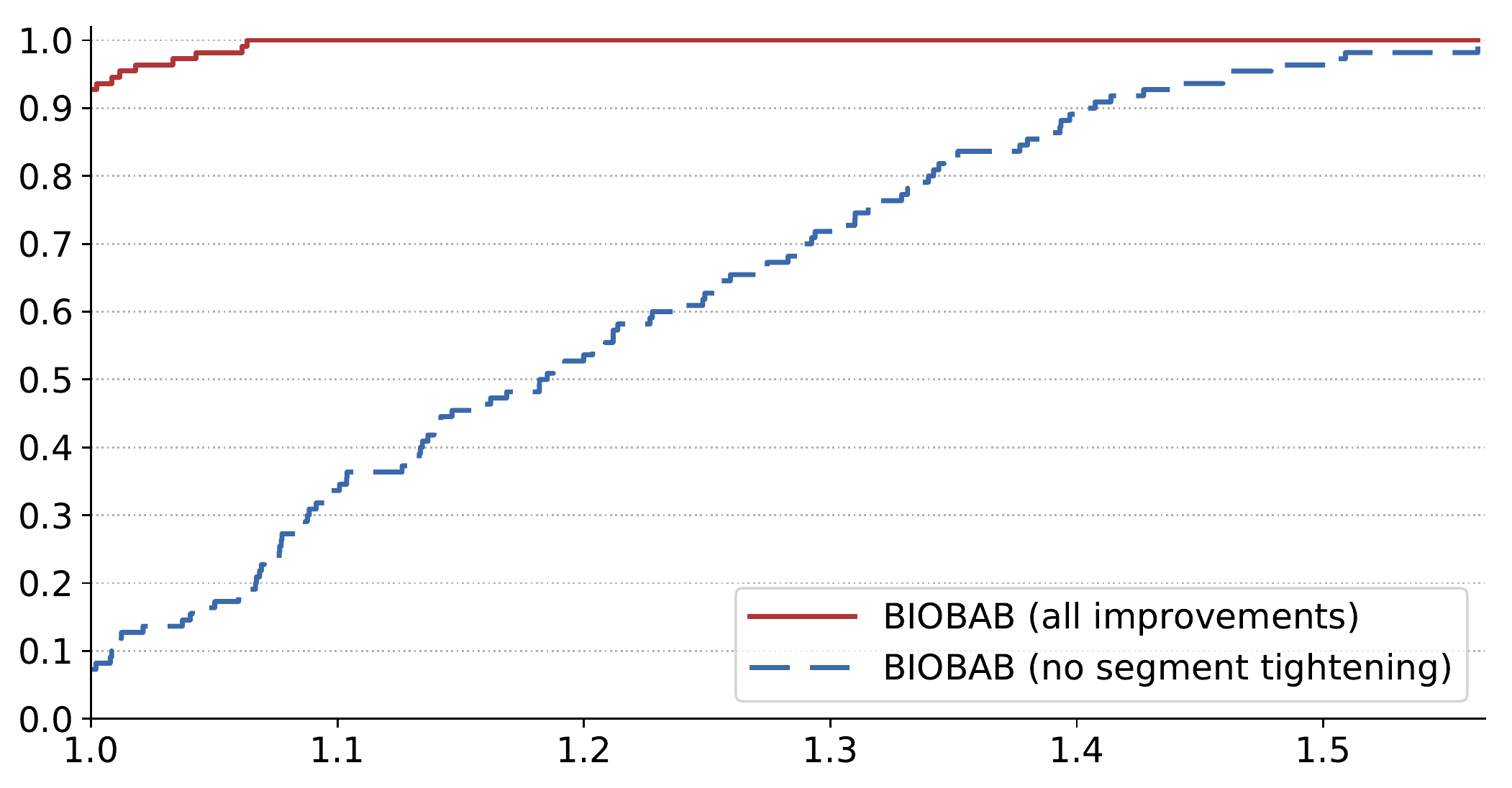}
    }
  \caption{Full BIOBAB vs no segment tightening}
  \label{fig:biobab-full-vs-nst}
\end{figure}
Approximately 70\% of UBOFLP instances cannot be solved within
two hours when segment tightening is disabled, which makes it a crucial
component. All SSUFLP Instances can be solved, but the profile of full BIOBAB
is still consistently better than the profile without segment tightening. Full
BIOBAB is up to 56\% faster.

We now also look at the performance profile of full BIOBAB, versus full BIOBAB
minus lower bound lifting, depicted in Figure~\ref{fig:biobab-full-vs-nll}. 
\begin{figure}
  \subfigure[UBOFLP]{
    \includegraphics[scale=.4]{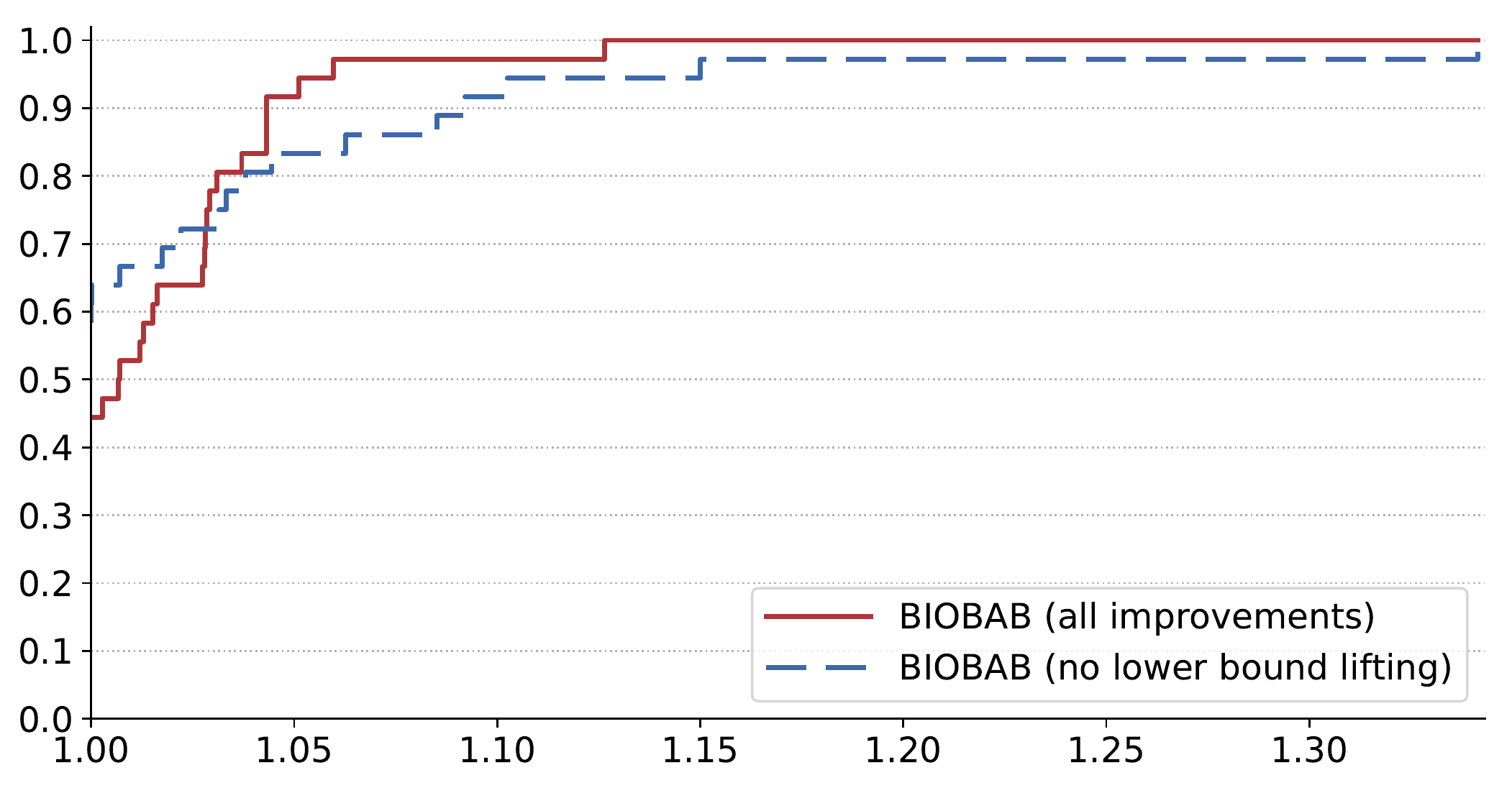}
  }
  \subfigure[SSUFLP]{
    \includegraphics[scale=.4]{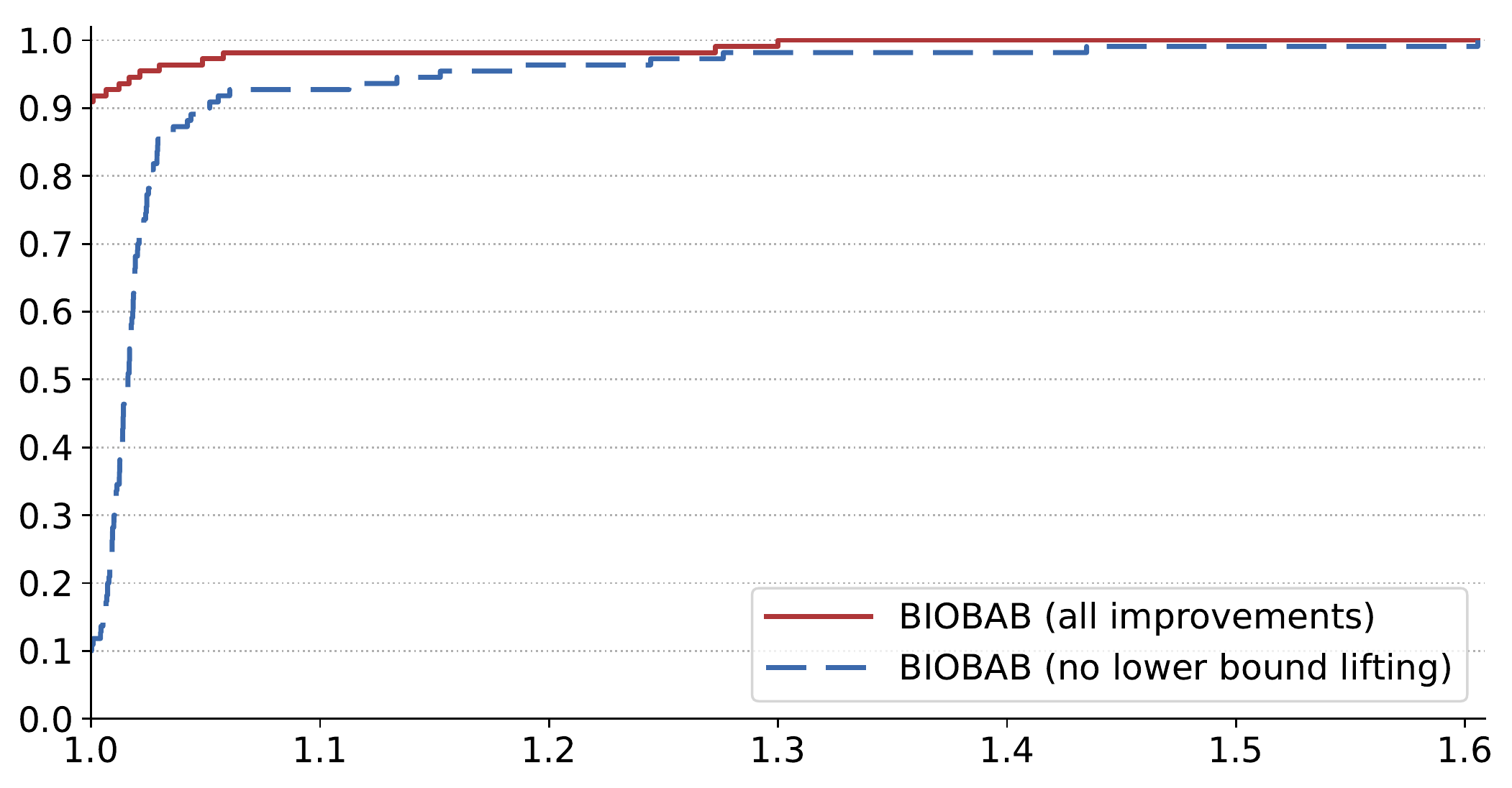}
  }
  \caption{Full BIOBAB vs no lower bound lifting}
  \label{fig:biobab-full-vs-nll}
\end{figure}
The impact of lower bound lifting is less spectacular, but still
significant. In rare cases not using lower bound lifting actually improves
performance, but only very slightly (less than 5\%, on some UBOFLP
instances). Lower bound lifting speeds up the algorithm by up to 34\% on UBOFLP
instances and by up to 60\% on SSUFLP instances.
As mentioned earlier, lower bound lifting only involves constant time
operations, so the overhead is minimal and will never cause an important
decrease in performance, even when it is entirely useless. Therefore it is safe
to use it in general.

Full BIOBAB is compared to the variant without OSB in
Figure~\ref{fig:biobab-full-vs-nosb}.
\begin{figure}
  \subfigure[UBOFLP]{
    \includegraphics[scale=.4]{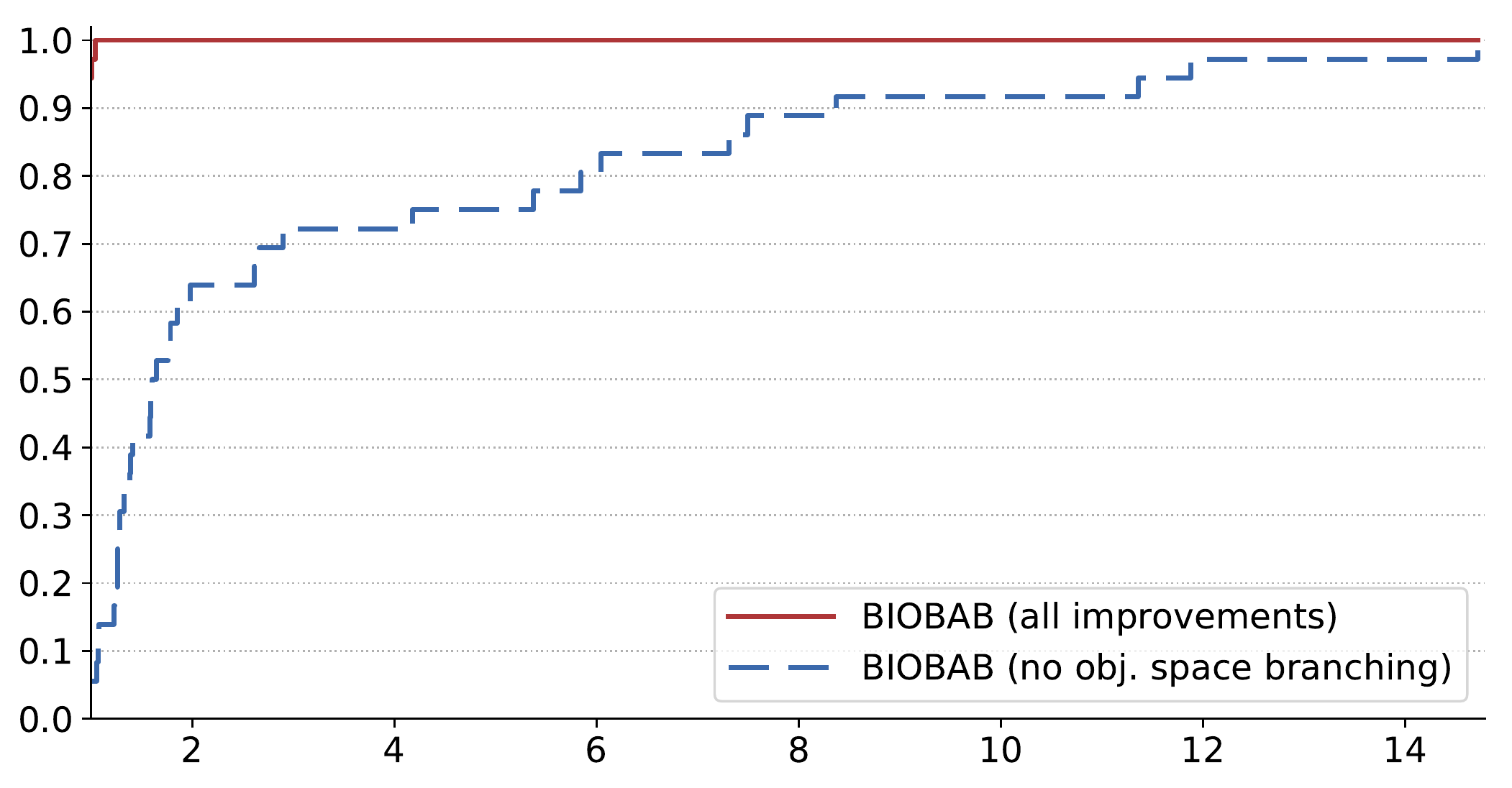}
  }
  \subfigure[SSUFLP]{
    \includegraphics[scale=.4]{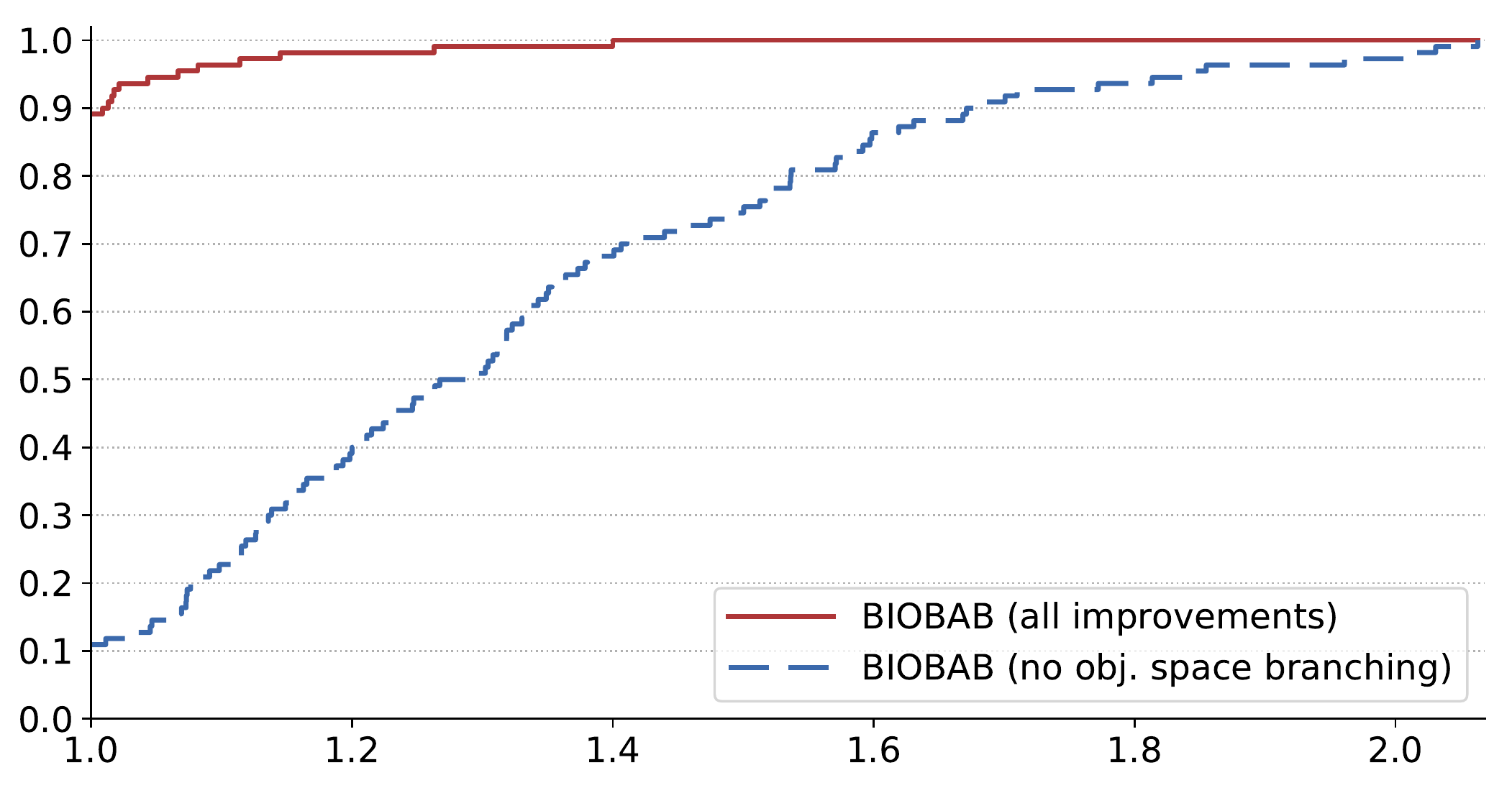}
  }
  \caption{Full BIOBAB vs no objective space branching}
  \label{fig:biobab-full-vs-nosb}
\end{figure}
The impact of objective space branching is important: using it
speeds up BIOBAB by a factor up to 14.7 for UBOFLP instances and up to 2.06 for
SSUFLP instances.
We note here that it is not surprising that objective space branching works
well when combined with segment tightening, since segment tightening creates
discontinuities in LB sets, thus allowing objective space branching to be
applied.

These performance profiles clearly show that every algorithmic improvement
brings added value, with rare exceptions for lower bound lifting.
In order to provide additional insight on the use of objective space branching,
we consider additional information related to full BIOBAB and full BIOBAB minus
OSB: we analyze CPU time, number of nodes explored during tree search and
number of linear programs solved during the whole search.
\begin{table}\scriptsize
\begin{center}
\begin{tabular}{lccccccc}
  \hline & \multicolumn{3}{c}{Full BIOBAB} && \multicolumn{3}{c}{No obj. space branching} \\
  \cline{2-4} \cline{6-8}
  Instance & \# LPs & \# nodes & CPU (s)& & \# LPs & \# nodes & CPU (s) \\
\hline
021 	&	23 	&	1 	&	0.30 	&&	23 	&	1 	&	0.28 \\
044 	&	59 	&	3 	&	0.57 	&&	59 	&	3 	&	0.61 \\
055 	&	73 	&	3 	&	0.89 	&&	73 	&	3 	&	0.85 \\
072 	&	95 	&	3 	&	1.59 	&&	95 	&	3 	&	1.56 \\
090 	&	135 	&	7 	&	3.04 	&&	197 	&	7 	&	4.18 \\
106 	&	157 	&	7 	&	5.07 	&&	235 	&	7 	&	6.33 \\
120 	&	173 	&	7 	&	6.44 	&&	267 	&	7 	&	8.88 \\
132 	&	193 	&	9 	&	8.30 	&&	335 	&	9 	&	13.52 \\
134 	&	197 	&	9 	&	8.58 	&&	339 	&	9 	&	15.16 \\
135 	&	199 	&	9 	&	8.41 	&&	333 	&	7 	&	14.72 \\
162 	&	263 	&	11 	&	16.78 	&&	401 	&	9 	&	26.39 \\
182 	&	466 	&	17 	&	32.53 	&&	699 	&	13 	&	51.58 \\
203 	&	478 	&	17 	&	41.57 	&&	697 	&	11 	&	66.53 \\
226 	&	255 	&	9 	&	31.41 	&&	319 	&	9 	&	44.19 \\
254 	&	259 	&	9 	&	41.72 	&&	335 	&	9 	&	53.46 \\
264 	&	269 	&	9 	&	47.21 	&&	345 	&	9 	&	58.12 \\
266 	&	271 	&	9 	&	49.48 	&&	347 	&	9 	&	63.44 \\
275 	&	277 	&	9 	&	53.35 	&&	303 	&	7 	&	56.18 \\
294 	&	285 	&	9 	&	64.19 	&&	361 	&	9 	&	84.88 \\
295 	&	277 	&	7 	&	71.97 	&&	363 	&	9 	&	88.12 \\
296 	&	289 	&	9 	&	74.64 	&&	365 	&	9 	&	94.16 \\
326 	&	361 	&	15 	&	112.16 	&&	1459 	&	29 	&	467.91 \\
342 	&	399 	&	19 	&	137.47 	&&	1125 	&	15 	&	397.68 \\
355 	&	473 	&	31 	&	165.05 	&&	3509 	&	61 	&	1379.35 \\
370 	&	455 	&	25 	&	195.29 	&&	983 	&	19 	&	386.76 \\
388 	&	461 	&	21 	&	204.75 	&&	3159 	&	39 	&	1494.41 \\
399 	&	531 	&	35 	&	236.88 	&&	1449 	&	33 	&	630.14 \\
410 	&	595 	&	45 	&	297.69 	&&	1061 	&	17 	&	551.82 \\
428 	&	645 	&	51 	&	342.32 	&&	7453 	&	89 	&	3885.49 \\
436 	&	609 	&	43 	&	351.86 	&&	1585 	&	31 	&	919.27 \\
449 	&	679 	&	53 	&	400.70 	&&	10586 	&	183 	&	5891.94 \\
458 	&	713 	&	57 	&	432.92 	&&	5387 	&	85 	&	3241.36 \\
472 	&	719 	&	59 	&	462.29 	&&	4681 	&	113 	&	2790.82 \\
482 	&	821 	&	75 	&	591.75 	&&	5511 	&	81 	&	3456.99 \\
499 	&	773 	&	67 	&	542.06 	&&	4209 	&	97 	&	2912.71 \\
500 	&	759 	&	65 	&	583.60 	&&	9003 	&	135 	&	6927.28 \\
\hline
\end{tabular}
\end{center}
\caption{Effect of objective space branching on tree search (UBOFLP instances).}
\label{tab:effect-of-osb}
\end{table}
Table~\ref{tab:effect-of-osb} shows the per-instance difference between the two
versions of BIOBAB, on UBOFLP instances. We note that the total number of nodes
explored does not vary that much
between both methods, but the number of LPs solved does: the version without
objective space branching solves more than ten times more LPs in some cases. We
believe this validates the idea that branching on objective space avoids
visiting several times regions that have already been established to be
dominated by the UB set; apparently visiting these regions requires solving
many LPs.

We now look closer at the number of branches generated by OSB. For that
purpose we count, for each instance solved, how many branches are generated by
OSB. We differentiate the number of branches generated at the root node and in
the remainder of the branch-and-bound tree. We report average values for these
two indicators, as well as the maximum tree depth at which OSB 
happened, in Table~\ref{tab:osb-data}. Instances that
are entirely solved without branching are disregarded.
\begin{table}
  \centering
  \bgroup
  \setlength\tabcolsep{20pt}
  \begin{tabular}{cccc}
    \hline
    Problem  & \makecell{\# root}
    & \makecell{\# other}
    & \makecell{max. depth} \\
    \hline
    SSUFLP & 23.84 & 2.07 & 12.24 \\
    UBOFLP & 5.37 & 0 & 0 \\
    \hline
  \end{tabular}
  \egroup
  \caption[caption]{Additional information on objective-space branching:
    average values.
    \\
    \# root: number of OSB branches at the root node\\
    \# other: number of OSB branches at other nodes\\
    max. depth: maximum depth at which OSB occurred\\
  }
  \label{tab:osb-data}
\end{table}
For both problem classes, OSB generates more branches at the root node, thus
indicating that integer solutions are already found while solving the root
node. For UBOFLP, OSB actually only happens at the root node; for SSUFLP, the
average number of branches at nodes other than the root node is very close to
2, which suggests that it typically happens when one new efficient point is
found. We also observe that a lower number of
branches at a given node does not necessarily indicate
that fewer integer solutions are found at this node, since segment tightening
can also eliminate the objective space between two integer points, thus
reducing the number of branches. But if there are $k$ OSB branches following
the bounding procedure at a given node, then we know that at least $k-1$
integer solutions were found during this bounding procedure.

\subsection{Comparison with criterion space search algorithms}
We implemented a
generic balanced box method as described in~\cite{Boland:2015} (it is sketched
in Algorithm~\ref{alg:box} in the Appendix), including solution harvesting and
other enhancements, within the same Python code base as BIOBAB. The frequency
with which solution harvesting is called is controlled using parameter
$\beta$. Like in the original paper, we set $\beta = 0.15$ following
preliminary experiments. Also like in the original paper, we find that in order
to produce lexicographic minima, solving
explicitly two successive IPs is more efficient than solving one weighted-sum
IP. Previously found integer solutions that are valid for the problem at hand
(i.e. within the same rectangle) are used to provide a cutoff value to the MIP
solver.

Within the same code base we also develop generic versions of the
$\epsilon$-constraint framework (see Algorithm~\ref{alg:epsi} in the Appendix)
as well as the bi-directional
$\epsilon$-constraint framework described in~\cite{Boland:2015}. We note here
that there are always two ways to run the single-directional
$\epsilon$-constraint algorithm, either starting with objective 1 or with
objective 2. For a given problem class, one of these two ways is typically
faster than the other, sometimes much faster. In these cases, it is likely that
the bi-directional version will not be as fast since it has to solve IPs in
the slower direction for 50\% of these IPs.
All $\epsilon$-constraint methods
benefit from the applicable enhancements developed for the balanced box method.

We first compare BIOBAB and other criterion space search methods on the
UBOFLP. We note here that BIOBAB solves the linear relaxation of the UBOFLP
to compute LB sets, whereas criterion space search methods call the MIP solver
directly.
Performance profiles are given in Figure~\ref{fig:uboflp-selected}.
\begin{figure}
  \centering
  \includegraphics[scale=.4]{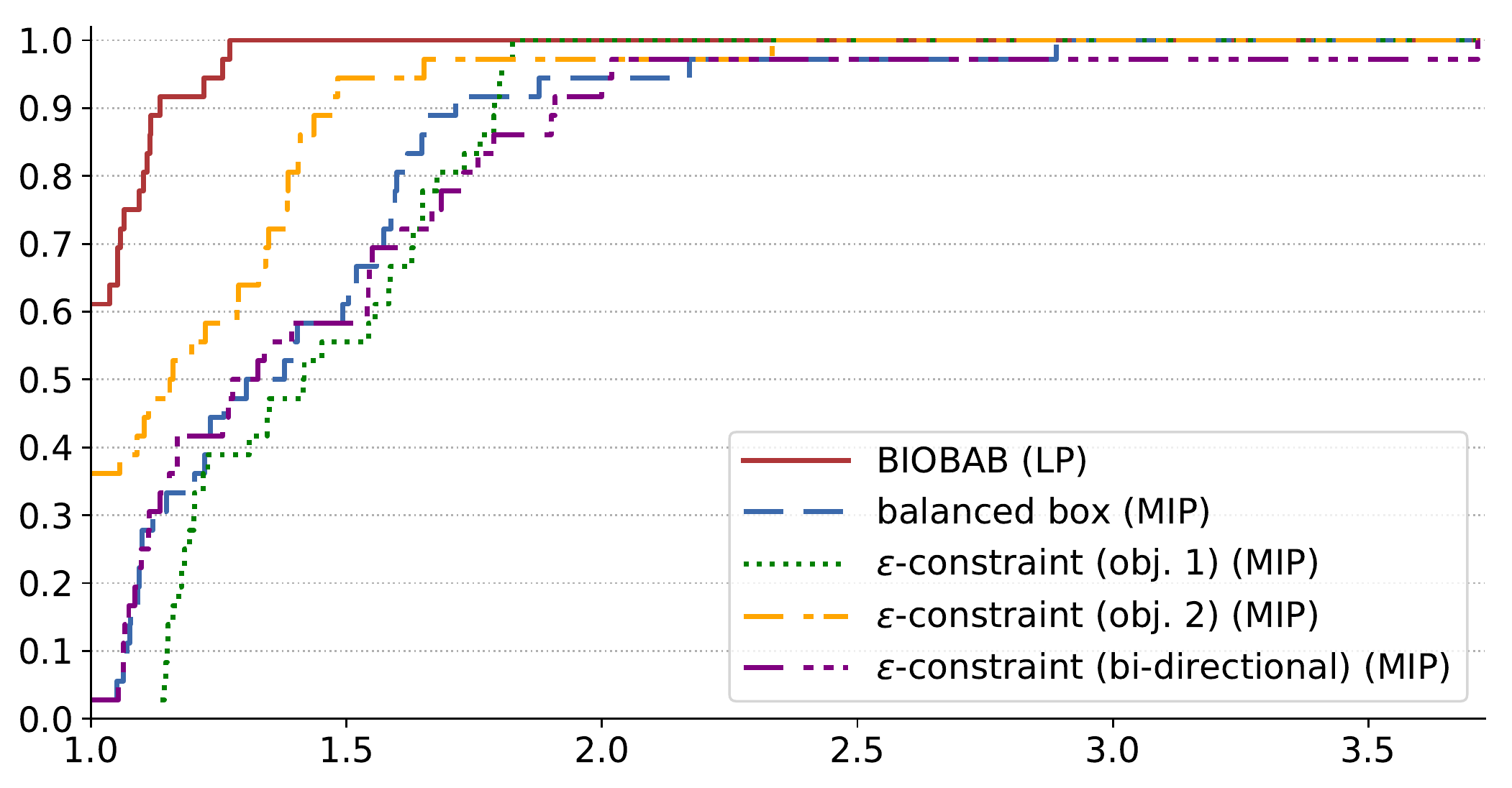}
  \caption{UBOFLP: BIOBAB vs criterion space search methods}
  \label{fig:uboflp-selected}
\end{figure}
BIOBAB appears to perform better than criterion space search methods on the
UBOFLP. It is however not overwhelming, as all methods always converge within a
factor 3 of the best performance overall, i.e. within the same order of
magnitude. Nonetheless, BIOBAB is definitely competitive on these instances. We
also note that single-directional $\epsilon$-constraint performs better than
its bi-directional counterpart.

We now look at the performance profiles of the same methods, this time on
SSUFLP instances. These profiles are represented in
Figure~\ref{fig:ssuflp-selected}.
\begin{figure}
  \centering
  \includegraphics[scale=.4]{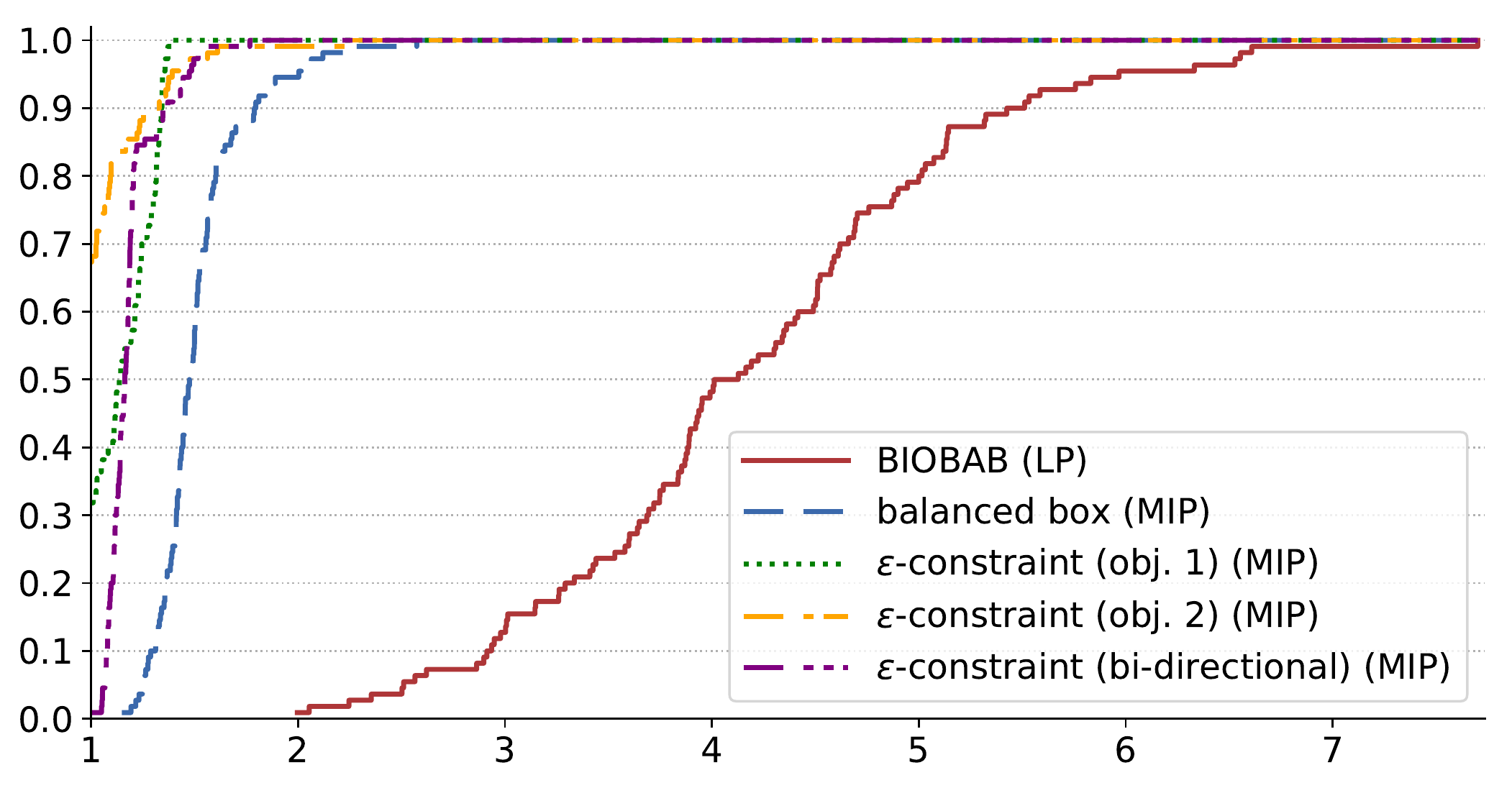}
  \caption{SSUFLP: BIOBAB vs criterion space search methods}
  \label{fig:ssuflp-selected}
\end{figure}
The picture is now different: BIOBAB is the worst of all compared methods. It
is however, still in the same order of magnitude (7 times slower than the best
method in the worst case). In this comparison, all criterion
space search methods benefit from the power of commercial MIP solvers, which
are especially efficient at tree exploration and branching on decision
variables. Even though we show above that objective space branching
significantly improves the performance of BIOBAB, branching on decision
variables is still performed in a somewhat naive way (fractional
variable with value closest to 1, breadth-first tree exploration). In the next
experiment we do not relax integrality constraints when computing LB sets for
the SSUFLP: IPs are solved within the bounding procedure of BIOBAB. An
expectation is that LB sets will be of better quality, at the cost of extra CPU
time. Better LB sets should reduce the size of the branch-and-bound tree of
BIOBAB, thus delegating part of the task of tree search to the commercial MIP
solver. Another expectation is that more integer solutions will be obtained
earlier in the search, which in turn will allow OSB and integer dominance to
speed up the search.
We now compare this version of BIOBAB to the same criterion space
search methods, in Figure~\ref{fig:ssuflp-MIP}.
\begin{figure}
  \centering
  \includegraphics[scale=.4]{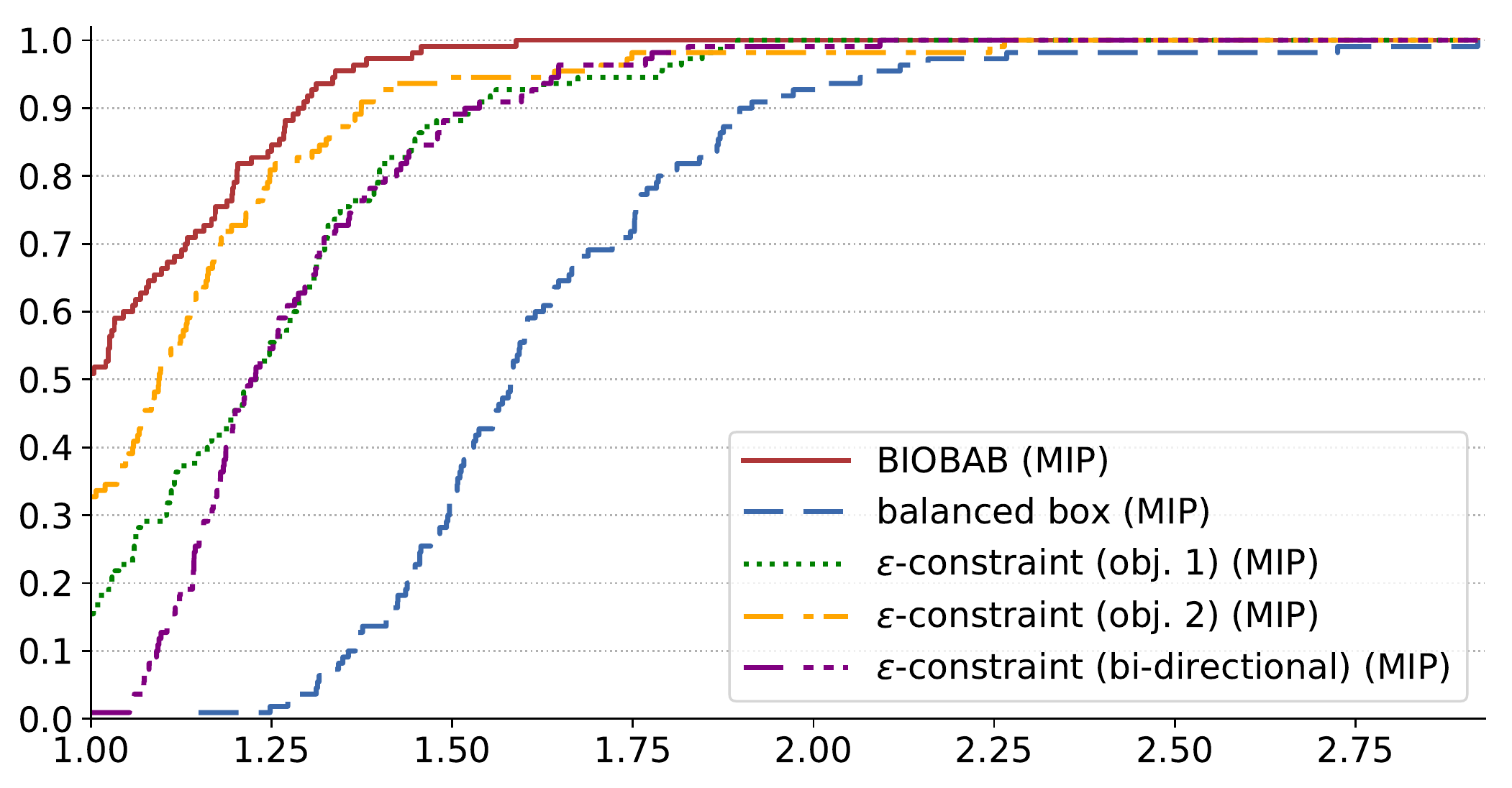}
  \caption{SSUFLP: BIOBAB (MIP) vs criterion space search methods}
  \label{fig:ssuflp-MIP}
\end{figure}
The performance profiles clearly show that using an IP instead of its linear
relaxation can be beneficial: BIOBAB is now performing much better,
actually slightly better than all criterion space search methods.

Additionally, we display in Figure~\ref{fig:ssuflp-nLPs} the total number of
linear programs (LPs) solved by each algorithm to generate the whole set of
non-dominated points. It is clear that the poor performance of BIOBAB using LPs
is correlated with the unusually high amount of LPs it requires to solve, while
BIOBAB using MIPs requires to solve a number of LPs comparable with that of
other algorithms. This strengthens the idea that efficient branching on
decision variables is the issue here.
\begin{figure}
  \centering
  \includegraphics[scale=.87]{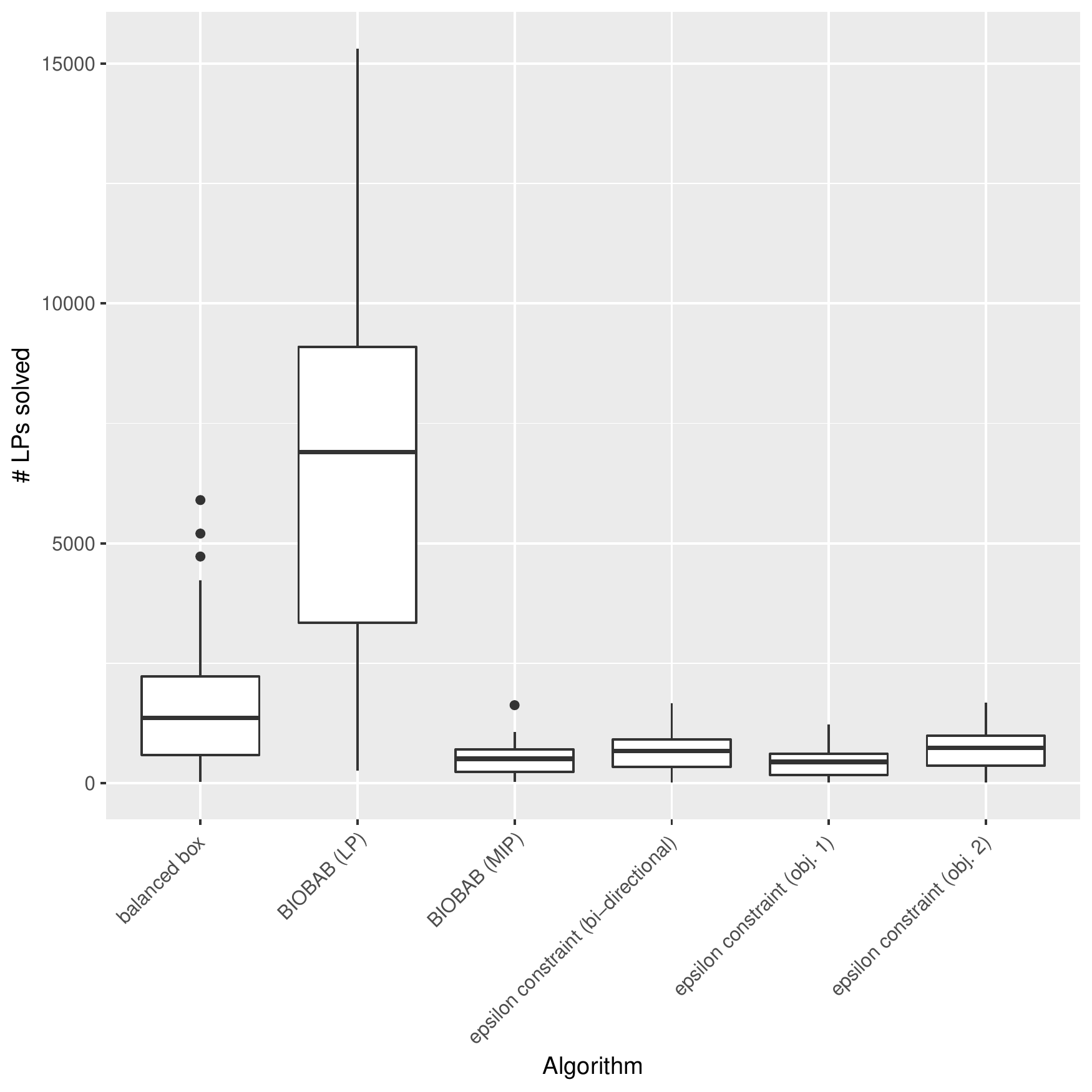}
  \caption{SSUFLP: number of LPs solved by each algorithm.}
  \label{fig:ssuflp-nLPs}
\end{figure}

\subsection{Comparison with bi-objective branch-and-bound}
We now compare BIOBAB with the branch-and-bound algorithm for bi-objective
mixed-integer programs from~\cite{Belotti:2013}. Their algorithm is for
mixed-integer programs, but they also provide results on the bi-objective set
covering problem, which only uses binary variables. We run BIOBAB on the same
instances and can therefore compare the CPU time required for solving these
instances. We present results using the same format as in their article. Their
computer is a 3.2 GHz workstation with 4 GB of RAM. 
We observe that BIOBAB is always faster, 
more than 10 times faster for the largest instances. Even if it is not clear
how much has to attributed to differences in CPU, we can safely conclude that
BIOBAB can compete with a state-of-the-art branch-and-bound method on these
instances.

\begin{table}
 \centering
  \bgroup
  \begin{tabular}{llll}
    \hline
    Instance size ($m \times n$) & \# instances & CPU (s) Belotti et al.
    & CPU (s) BIOBAB \\
    \hline
    5 $\times$ 25 & 1 & 0.1 & 0.01 \\
    10 $\times$ 50 & 1 & 2.1 & 0.73 \\
    10 $\times$ 100 & 4 & 49.3 & 3.76 \\
    \hline
  \end{tabular}
\egroup
  \caption[caption]{Comparison of BIOBAB and bi-objective branch and bound
    from~\cite{Belotti:2013} on bi-objective set covering problem instances:
    average CPU time. Instances have $m$ variables and $n$ constraints.
    \\
    Belotti et al.: 3.2 GHz workstation with 4 GB of RAM.\\
    BIOBAB: 2.6 GHz Xeon CPU with a 4 GB RAM limit.
  }
\end{table}

\subsection{Application to cases which cannot benefit from the
 power of MIP solvers}
\label{sec:bitoptw}
A major strength of state-of-the-art MIP solvers lies in their ability to
conduct efficient branch-and-bound or branch-and-cut tree search, through
involved tree exploration heuristics that are industrial secrets. However, 
it is not always possible to benefit from
these advanced methods.
Sometimes even the single objective problem version is too hard to
solve with a MIP solver. We want to assess how BIOBAB fares in such
cases. 
For that purpose, we use the BITOPTW
as test case.
The BITOPTW is formally described in Appendix~\ref{sec:bitoptw-compact}. Such a
model is challenging for current state-of-the-art MIP solvers.

\subsubsection{Lower bound set: column generation}

Since state-of-the-art exact methods for single-objective routing problems
mostly rely on column generation based techniques \citep[see,
e.g.][]{Baldacci:2012survey}, we also generate lower bound sets for the BITOPTW
by means of column generation.

Let $P_r$ denote the total score or profit achieved by route $r$, $C_r$ the
total travel cost of route $r$, and let $a_{ir}$ indicate whether location $i$
is
visited by route $r$ ($a_{ir} = 1$) or not ($a_{ir} = 0$). Using binary
variables $x_{r}$ equal to $1$ if route $r$ is selected from the set $\Omega$
of all feasible routes, the BITOPTW can also be
formulated as a path-based model: 
\begin{align}
	f_1(x) = \min \: \sum_{r \in \Omega} C_r x_r \label{btopP:OFcost}\\
	f_2(x) = \max \: \sum_{r \in \Omega} P_r x_r \label{btopP:OFprofit}
\end{align}
\begin{align}
	\sum_{r \in \Omega} a_{ir} x_r & \leq 1
	&& \forall i \in N  
	\label{btopP:covering} \\
	\sum_{r \in \Omega} x_r & = m
	 \label{btopP:fleet}\\
	x_r & \in \{0,1\} 
	&& \forall r \in \Omega  
	\label{btopP:binary}.
\end{align}
Relaxing integrality requirements on the $x_r$ variables, we replace
constraints \eqref{btopP:binary} with:
\begin{align}
	x_r & \geq 0 && \forall r \in \Omega \label{btopP:binaryrealx}
\end{align}
We also combine the two objective functions into a weighted sum ($w_1$ and $w_2$
giving the respective non-negative weights determined by our LB set computation procedure $bound$):
\begin{align}
	\min \: \sum_{r \in \Omega} ( w_1 C_r - w_2 P_{r} ) x_r
        \label{eq:weighted-objective}
\end{align}
We thus obtain a single objective linear problem that can be solved by means of
column generation, which allows us to compute LB sets
using Algorithm~\ref{alg:aneja}.
In column generation \citep[see, e.g.][]{desrosiers2005primer}, in each iteration a subset of promising columns is
generated and appended to the restricted set of columns
$\Omega^{\prime}$ and the single objective linear program is re-solved on
$\Omega^{\prime}$. Column generation continues as long as new promising columns
exist. Otherwise, an optimal solution has been found (i.e. in that case, an
optimal solution of the single objective linear program on $\Omega^{\prime}$ is
also an optimal solution of the single objective linear program on $\Omega$).
Promising columns are identified using dual information. 
Let $\pi_i$ denote the dual variable associated with
constraint~\eqref{btopP:covering} for a given $i$ and $\alpha$ the dual
variable associated with constraint~\eqref{btopP:fleet}, the pricing subproblem
we have to solve corresponds to:
\begin{align}
\min \: w_1 C_r - w_2 P_r - \sum_{i \in N} a_{ir} \pi_i - \alpha 
\end{align}
subject to constraints~\eqref{bitoptw:start}--\eqref{bitoptw:zbin} given in Appendix~\ref{sec:bitoptw-compact}, omitting
subscript $k$. It is an elementary shortest path problem with resource
constraints that can be solved by
means of a labeling algorithm \citep[cf.][]{feillet:2004exact}. 
In our labeling algorithm, a label carries the following information: the node the label is associated with, the time consumption until that node, the reduced cost so far, which nodes have been visited along the path leading to the node, and a pointer to the parent label. In order to compute the reduced cost of the path associated with a given label, we use a reduced cost matrix that is generated before the labeling algorithm is called. The reduced cost of arc $(i,j)$ is $\bar{c}_{ij} = w_1 c_{ij} - w_2 S_{i} - \pi_i $, where if $i = 0$, $\pi_i$ is replaced by $\alpha$ and $c_{ij}$ gives the costs for traversing arc $(i,j)$ and $S_i$ the score for visiting location $i$.
Since
the aim of this paper is not to investigate the most efficient pricing
algorithm, we refrain from adopting all enhancements proposed in the literature
\citep[e.g.][]{righini:2009decremental}. 

Note that during the execution of BIOBAB we keep all previously generated
columns in the column pool. We only temporarily
deactivate those columns that are incompatible with current branching
decisions. When branching on objective space, it can happen that the current
pool of columns does not allow
to produce a feasible solution. This can be due to two reasons: either (i)
because there exists no feasible solution to the current problem or (ii)
because there exists a feasible solution but it cannot be reached with the
currently available columns. In order to fix this issue and guarantee
feasibility, we use a dummy column which allows to satisfy every constraint
from the current problem, including the branching decisions described in
Section~\ref{sec:bap-branching}. These include branching on control
points. Therefore, some control points may be mandatory and others may be
forbidden. Thus, the dummy column covers all mandatory control points, does not
cover any forbidden control point, uses all the available 
vehicles, has a cost inferior to the maximum allowed cost and a profit superior
to the minimum allowed profit. Using this column is penalized in the objective
function, so that the column generation procedure converges to feasible
solutions that do not use the dummy column. Assuming the variable associated to
this dummy column is $x_D$, the objective function described in
Equation~\eqref{eq:weighted-objective} is modified as follows:

\begin{align}
	\min \: \sum_{r \in \Omega} ( w_1 C_r - w_2 P_{r} ) x_r + M x_D,
        \label{eq:penalized-objective}
\end{align}
where $M$ is an arbitrarily large number. The dummy column is only activated
when no feasible solution can be found, and it is systematically deactivated
after column generation converges. At this point, if $x_D$ has a strictly
positive value then there does not exist a feasible solution that satisfies all
the branching decisions.

\subsubsection{Branching scheme}\label{sec:bap-branching}
Objective space branching is performed as usual (see
Section~\ref{sec:branching}).
Since objective space branching involves setting bounds on both objectives, we
include these constraints right from the start (initially setting the respective
bounds to infinity) and later update them according to the branching
decisions:
\begin{align}
	\sum_{r \in \Omega} C_r x_r & \leq c_1, \label{con:obj1}\\
	\sum_{r \in \Omega} - P_r x_r & \leq c_2, \label{con:obj2} 
\end{align}
where $c= (c_1, c_2)$ is the local nadir point. 
  In order to properly incorporate dual information from these two
constraints into the subproblem, we modify the reduced cost matrix. 
Let $\lambda_1$ be the dual variable value associated with constraint \eqref{con:obj1} and $\lambda_2$ the value associated with constraint \eqref{con:obj2}, then the reduced cost of a given arc $(i,j)$ is given by $\bar{c}’_{ij} = (w_1 - \lambda_1) c_{ij} - (w_2 + \lambda_2) S_i - \pi_i$.
Regarding decision-space branching, we either branch on control points or on
arcs (Line~\ref{line:psb} in Algorithm~\ref{alg:branch}). First, we check if a
control point is visited a fractional number of times. Since each lower bound
represents a set of solutions, the number of times a control point is visited
may take different values within the same LB set. In order to select a control
point to branch on, we consider the supported (fractional) solutions defining
the current LB set. We then 
average the number of visits for each control point over this set of
solutions. The control point with the average number of visits closest to $0.5$ is
then selected for branching. In the case where each control point is visited 0
or 1 time on average, we check for arcs that are traversed a fractional number
of times on average,  following the same procedure. 
Branching on control points can be achieved by modifying constraint~\eqref{btopP:covering} in the master problem. In order to force the visit of control point $i$, this constraint becomes $\sum\limits_{r \in \Omega} a_{ir}x_r \geq 1$. In order to forbid the visit of $i$, this constraint becomes $\sum\limits_{r \in \Omega} a_{ir}x_r \leq 0$.
As is usual in
branch-and-price for routing, branching on arcs involves adding constraints to the subproblem.

\subsubsection{Applying the $\epsilon$-constraint method}

In order to asses the quality of our BIOBAB on the BITOPTW, we compare it to
the $\epsilon$-constraint method, using the faster direction. In this case, the
function $solveMIP$ (Line~\ref{line:weightedsum} of Algorithm~\ref{alg:epsi} in
the Appendix) uses the same tree search algorithm as BIOBAB, i.e.~the same
rules for branching are applied, with the exception of objective space
branching, since a single-objective problem is solved.
The bounding procedure is similar to the one used in BIOBAB, i.e. the same
column generation code is used. However, instead of solving a succession of LPs
to compute a LB set,
a single solution is produced in
each call to the procedure. Other than that, the code is exactly the same for
both methods.
This means that in the $\epsilon$-constraint framework, previously generated
columns are reused. Additionally, a single-objective solution to the linear
relaxation yields a LB segment
; therefore, the
existing UB set can be used to fathom subtrees using the filtering procedure
described in Section~\ref{sec:filter}. This works in combination with the fact
that all produced integer solutions are stored, thus enhancing the
$\epsilon$-constraint implementation with improvements developed for
BIOBAB.

\subsubsection{BITOPTW Benchmark instances}

We use the TOPTW instances
by~\cite{righini:2009decremental} and we reduce their size by considering
the first 15, 20, 25, 30 and 35 customers only. Using instances c101\_100,
r101\_100, rc101\_100 and pr01 and the number of customers mentioned above, we
generate test instances for 1, 2, 3 and 4 vehicles. In total there are 80
BITOPTW instances.

\subsubsection{Computational results}
We now compare BIOBAB with the $\epsilon$-constraint method.
Column generation can be time-intensive, so in this case all procedures are
implemented in C++. The same column generation code is used by both BIOBAB and
$\epsilon$-constraint. We run both algorithms on the same computer as mentioned
above, and produce the performance profile depicted in Figure~\ref{fig:cg}.
\begin{figure}
  \centering
  \includegraphics[scale=.4]{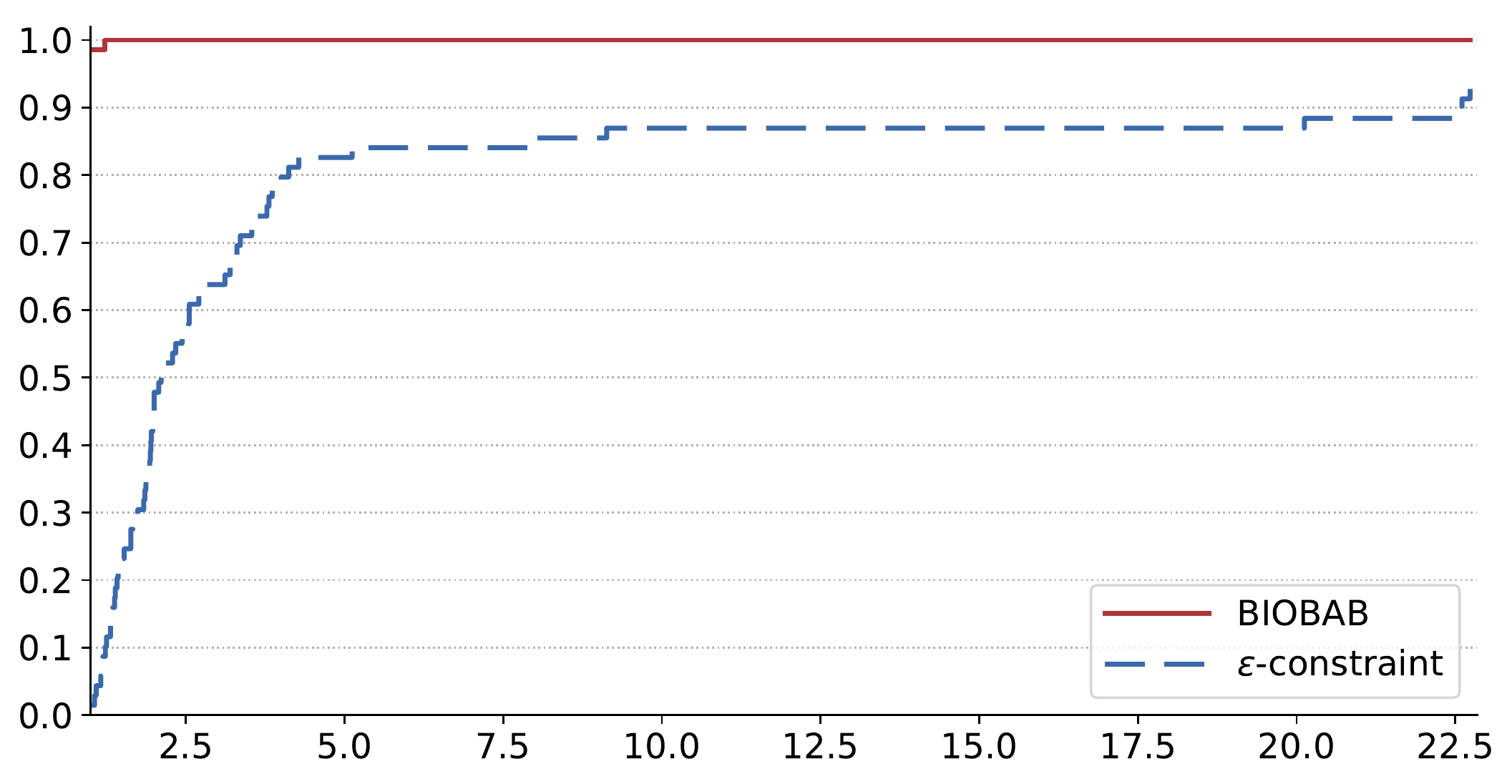}
  \caption{BIOBAB vs $\epsilon$-constraint when using column generation to
    solve the underlying linear program.}
  \label{fig:cg}
\end{figure}
As we can see, BIOBAB offers a better performance overall. Although the
$\epsilon$-constraint method is within the same order of magnitude for almost
90\% of the instances, which is good, it can get up to 20 times slower in some
rare cases. More strikingly, there are instances solved by BIOBAB that
cannot be solved within the allotted CPU budget by the $\epsilon$-constraint
method.

\section{Conclusion}\label{sec:conclusion}
We have introduced a new algorithm for bi-objective integer optimization. 
The algorithm is a generalization of
single-objective branch-and-bound
to the bi-objective
context. 
Our algorithm is generic and
uses a branching rule  exploiting the bi-objective nature of the underlying optimization problem.
We have also
developed several enhancements based on the integrality of the objective values of
efficient solutions. Experiments show that these enhancements significantly
improve the performance of BIOBAB, and make it competitive with
state-of-the-art generic methods for bi-objective optimization for integer
programs.
In order to demonstrate the versatility of our framework, we have also applied
it to a routing problem where the lower bound linear program is solved by means
of column generation, thus providing, to the best of our knowledge, the first
bi-objective branch-and-price algorithm. Standard criterion space search
methods rely on the efficiency of MIP solvers and show their limitations in
such cases:
the proposed branch-and-bound algorithm
outperforms the $\epsilon$-constraint method, which is a
standard benchmark for multi-objective optimization. 

One research
perspective lies in generalizing the concept of lower bound set to more than
two objectives, thus allowing exact approaches relying on this notion for more
than two objectives. However this will raise the issue of the cardinality of the
efficient set, which typically increases exponentially with the number of
objectives. Therefore another research perspective lies in deriving
approximation methods for multi-objective optimization based on exact
procedures, with the goal of sacrificing neither the quality of solutions nor
the readability of the produced set of solutions.

\section*{Acknowledgments}
We wish two thank four anonymous referees for their valuable
comments. Furthermore, financial support from the Austrian Science Fund (FWF):
P23589-N13 is gratefully acknowledged.

\appendix

\label{appendix}

\section{Algorithms for criterion space search methods}

Algorithms~\ref{alg:epsi} and~\ref{alg:box} outline the $\epsilon$-constraint method and the balanced box method, respectively. The set $P$ denotes the set of Pareto optimal solutions.
 Function $solveMIP$ can either be a call to Gurobi or CPLEX or to a
 tailor-made branch-and-bound algorithm. The value of $\epsilon$ is supposed to
 be valid for the problem at hand; for instance with pure integer problems with
 integer coefficients for the objective function, $\epsilon= 1$ is such a valid
 value.
\begin{algorithm}
  \caption{$\epsilon$-constraint algorithm}
  \begin{algorithmic}
    \STATE $P \leftarrow \emptyset$
    \STATE $\epsilon$-$constraint$ $\leftarrow f_2 \leq \infty$
    \WHILE {MIP is feasible}
    \STATE $x \leftarrow solveMIP(lexmin(f_1, f_2), \epsilon$-$constraint)$
    	\label{line:weightedsum}
    	\STATE $P \leftarrow P \cup x$
    \STATE $\epsilon$-$constraint$ $\leftarrow f_2 \leq f_2(x) - \epsilon$
   	 \label{line:epsi-constr}
    \ENDWHILE
    \RETURN $P$
  \end{algorithmic}
  \label{alg:epsi}
\end{algorithm}

\begin{algorithm}
  \caption{Balanced box method}
  \begin{algorithmic}
    \STATE $P \leftarrow \emptyset$, $Rectangles \leftarrow \emptyset$
   \STATE $x^T \leftarrow solveMIP(lexmin(f_1, f_2))$
   \STATE $x^B \leftarrow solveMIP(lexmin(f_2,f_1))$
    \STATE $P \leftarrow P \cup \{ x^T, x^B \}$	
    \STATE $Rectangles \leftarrow Rectangles \cup \{ rectangle \: z^T=(f_1(x^T),f_2(x^T)), z^B = (f_1(x^B), f_2(x^B)) \}$
    \WHILE {$Rectangles \neq \emptyset$}
    \STATE $Rectangles.pop(rectangle(z^1, z^2))$
    \STATE $rectangle^B \leftarrow rectangle((z^1_1,(z^1_2+z^2_2)/2),z^2)$
    \STATE $\bar{x}^1 \leftarrow solveMIP( lexmin(f_1,f_2), rectangle^B constraints)$
    \STATE $\bar{z}^1 =( f_1(\bar{x}^1), f_2(\bar{x}^1))$
    \IF {$ \bar{z}^1 \neq z^2$} 
    	\STATE $P \leftarrow P \cup \{ \bar{x}^1 \}$
    	\STATE $Rectangles \leftarrow Rectangles \cup rectangle( \bar{z}^1,z^2)$
	\ENDIF 
  \STATE $rectangle^T \leftarrow rectangle(z^1, (\bar{z}^1_1-\epsilon, (z^1_2+z^2_2)/2))$ 
   \STATE $\bar{x}^2 \leftarrow solveMIP( lexmin(f_2, f_1), rectangle^T constraints)$
   \STATE $\bar{z}^2 =( f_1(\bar{x}^2), f_2(\bar{x}^2))$
   \IF {$ \bar{z}^2 \neq z^1$} 
   	\STATE $P \leftarrow P \cup \{ \bar{x}^2 \}$
    	\STATE $Rectangles \leftarrow Rectangles \cup rectangle( z^1,\bar{z}^2)$
	\ENDIF 
    \ENDWHILE
    \RETURN $P$
  \end{algorithmic}
  \label{alg:box}
\end{algorithm}

\section{UBOFLP: Model formulation}
\label{sec:uboflp-model}

In the UBOFLP, we are given a set $V$ of potential facilities and a set $N$ of locations. Furthermore, we denote by $N_i$ the set of locations that can be covered by facility $i$ because they are within a certain radius. Each facility has opening costs $F_i$ and each location has a weight or demand $W_j$. The considered objectives simultaneously minimize the opening costs of facilities and maximize the total covered demand.

Using binary decision variables $y_{i} \in \{0,1\}$ equal to $1$ if facility $i$ is opened and $0$ otherwise, and $x_{ij} \in \{0,1\}$ equal to $1$ if location $j$ is assigned to facility $i$, we formally define the UBOFLP as follows:
\begin{align}
	\min \: &\sum_{i \in V}  F_{i} y_{i}
		\label{uboflp:OFcost}\\
	\max \: &\sum_{i \in V} \sum_{j \in N_i} W_j x_{ij}
		\label{uboflp:OFcoverage}
\end{align}
subject to:
\begin{align}
	x_{ij} & \leq y_i 
	&& \forall i \in V, j \in N, 
	\\
	\sum_{i \in V} x_{ij} & = 1
	&& \forall j \in N,
	\\ 
	y_i  & \in \{0,1 \}
	&& \forall i \in V, 
	\\
	 x_{ij} & \in \{0,1 \}
	 && \forall i \in V, j \in N.
\end{align}

\section{SSUFLP: Model formulation}
\label{sec:ssuflp-model}

In the SSUFLP, we are given a set $V$ of potential facilities and a set $N$ of locations that have to be assigned to one of the facilities each. Each facility has opening costs $F_i$ and it costs $c_{ij}$ to assign location $j$ to facility $i$. The considered objectives simultaneously minimize the total opening costs of facilities and the total assignment costs.

Using binary decision variables $y_{i} \in \{0,1\}$ equal to $1$ if facility $i$ is opened and $0$ otherwise, and $x_{ij} \in \{0,1\}$ equal to $1$ if location $j$ is assigned to facility $i$, we formally define the SSUFLP as follows:
\begin{align}
	\min \: &\sum_{i \in V}  F_{i} y_{i}
		\label{ssuflp:OFfixedcost}\\
	\min \: &\sum_{i \in V} \sum_{j \in N} c_{ij} x_{ij}
		\label{ssuflp:OFassigncost}
\end{align}
subject to:
\begin{align}
	x_{ij} & \leq y_i 
	&& \forall i \in V, j \in N, 
	\\
	\sum_{i \in V} x_{ij} & = 1
	&& \forall j \in N,
	\\ 
	y_i  & \in \{0,1 \}
	&& \forall i \in V, 
	\\
	 x_{ij} & \in \{0,1 \}
	 && \forall i \in V, j \in N.
\end{align}

\section{BITOPTW: Model formulation}
\label{sec:bitoptw-compact}
The BITOPTW is defined on a directed graph $G= (V,A)$, where $A$
is the set of arcs and $V$ the set of vertices, representing the starting
location (vertex $0$), the ending location (vertex $n+1$) and $n$ control
points. Each control point $i$ is associated with a score
$S_i$, a service time $d_i$ and a time window $[e_i,l_i]$. Each route $k \in
K$, with $|K| = m$,
has to start at location $0$ and end at location $n+1$ and
each arc $(i,j)$ is associated with travel cost $c_{ij}$ and travel time
$t_{ij}$. The aim is to maximize the total collected score and to
simultaneously minimize the total travel cost.
Using binary decision variables $z_{i} \in \{0,1\}$ equal to $1$ if  location $i$ is visited and $0$ otherwise, and $y_{ijk} \in \{0,1\}$ equal to $1$ if arc $(i,j)$ is traversed by route $k$ and $0$ otherwise,  and continuous variables $B_{ik}$, denoting the beginning of service at $i$ by route $k$, we formally define the BITOPTW as follows:
\begin{align}
	\min \: &\sum_{k \in K} \sum_{(i,j) \in A} c_{ij} y_{ijk}
		\label{bitoptw:OFcost}\\
	\max \: &\sum_{i \in V \setminus \{0,n+1\}} S_i z_{i}
		\label{bitoptw:OFprofit}
\end{align}
subject to:
\begin{align}
	\sum_{j \in V \setminus \{0\}} y_{0jk} &= 1 
	&& \forall k \in K
		\label{bitoptw:start}\\
	\sum_{i \in V \setminus\{n+1\}} y_{i,n+1,k} &= 1 
	&& \forall k \in K 
		\label{bitoptw:end}\\
	 \sum_{k \in K} \sum_{j \in V \setminus\{n+1\}}  y_{jik}  &= z_{i}
	 && \forall i \in V \setminus \{0,n+1\}
	 	\label{bitoptw:visit}\\
	\sum_{j \in V \setminus\{n+1\}} y_{jik} - \sum_{j \in V \setminus \{0\}} y_{ijk} & = 0
	&& \forall k \in K, i \in V \setminus \{0,n+1\}
		\label{bitoptw:inout}\\
	(B_{ik} + d_i + t_{ij}) y_{ijk} &\leq B_{jk}
	&& \forall k \in K, (i,j) \in A
		\label{bitoptw:time} \\
	e_i \leq B_{ik} & \leq l_i
	&& \forall k \in K, i \in V
		\label{bitoptw:tw}\\
	y_{ijk} & \in \{0,1\} 
	&& \forall k \in K, (i,j) \in A
		\label{bitoptw:ybin}\\
	z_{i} & \in \{0,1\}
	&& \forall i  \in V \setminus \{0,n+1\}
		\label{bitoptw:zbin}
\end{align}
Objective function~\eqref{bitoptw:OFcost} minimizes the total
routing costs while objective function~\eqref{bitoptw:OFprofit} maximizes the
total collected profit. Constraints~\eqref{bitoptw:start} and
~\eqref{bitoptw:end} make sure that each route starts at the defined starting
point and ends at the correct ending point.
Constraints~\eqref{bitoptw:visit} link the binary decision variables and~\eqref{bitoptw:inout} ensure connectivity for visited nodes.
Constraints~\eqref{bitoptw:time} set the time variables and~\eqref{bitoptw:tw}
make sure that time windows are respected. We note that constraints~\eqref{bitoptw:time} are not linear but they can easily be linearized using big $M$ terms.

\end{document}